\documentclass[sts]{imsart}

\pdfoutput=1

\RequirePackage{amsthm,amsmath,amsfonts,amssymb}
\RequirePackage[authoryear,sort&compress]{natbib}
\RequirePackage{graphicx}
\graphicspath{{figures/}}
\RequirePackage{booktabs}
\RequirePackage{array}
\RequirePackage{url}
\RequirePackage{xcolor}

\InputIfFileExists{preprint.cfg}{}{}
\makeatletter
\ifdefined\PREPRINT\def\journal@name{}\fi
\makeatother

\startlocaldefs
\newcommand{\R}{\mathbb{R}}
\newcommand{\Sig}{\Sigma}
\newcommand{\F}{\mathcal{F}}
\newcommand{\B}{\mathcal{B}}
\newcommand{\Sset}{\mathcal{S}}

\newcommand{\trans}{^{\top}}
\newcommand{\lb}{\ell}
\newcommand{\ub}{u}

\theoremstyle{plain}
\newtheorem{theorem}{Theorem}
\newtheorem{proposition}{Proposition}
\newtheorem{corollary}{Corollary}
\newtheorem{lemma}{Lemma}
\theoremstyle{definition}
\newtheorem{remark}{Remark}
\endlocaldefs

\begin{document}

% Number sections and subsections only; \paragraph headers stay unnumbered
% run-in.
\setcounter{secnumdepth}{2}

\begin{frontmatter}
\runtitle{One Curve, Two Literatures}
\runauthor{T. Schmelzer and T. Hastie}

\begin{aug}
\author[A]{\fnms{Thomas}~\snm{Schmelzer}\ead[label=e1]{thomas@jqr.ae}\orcid{0009-0009-9771-3501}}
\author[B]{\fnms{Trevor}~\snm{Hastie}\ead[label=e2]{hastie@stanford.edu}\orcid{0000-0002-0164-3142}}

\address[A]{Jebel Quant Research, Abu Dhabi\printead[presep={\ }]{e1}.}
\address[B]{Department of Statistics, Stanford University\printead[presep={\ }]{e2}.}
\end{aug}

\title{The Critical Line Algorithm and the Constrained LASSO:\\ One Curve, Two Literatures\\[8pt]
{\normalsize\spaceskip=0pt\relax\mbox{Working note, circulated for comment}}}

\begin{abstract}

Many statistical procedures compute an entire solution path rather than a single
estimator. For one class the path is piecewise linear, traced corner to corner by an
active-set homotopy. Two of its members coincide exactly: under $\Sigma=X^\top X$ and
$\mu=X^\top y$ the gross-exposure-constrained mean--variance program and the
constrained LASSO trace the same piecewise-linear curve, and they keep doing so under
arbitrary linear equality and inequality constraints. Which parameter is swept, a leverage
budget or a return tilt, is not a choice the curve notices: under a homogeneous mandate
the two sweeps differ by a scalar. A mandate that holds the portfolio invested costs a
radial rescaling instead, and the efficient frontier is that curve rescaled; its corners
pass unnoticed because the frontier is continuously differentiable where they fall. We
give the map between the two parametrisations and the single way it degenerates.
That mean--variance selection and the LASSO instantiate one parametric quadratic program
is prior art (G\"artner, Jaggi and Maria, 2012). The Critical Line Algorithm itself lies
off that route, and the correspondence established here is new.

The identity is one of curves, not of statistical experiments. The map between the
parametrisations is computed from the data, so the geometry of the path transfers while
quantities averaged over the response do not. We cross that boundary deliberately and
give the degrees of freedom of the constrained fit under arbitrary linear constraints. In
the long-only, fully invested case it is the expected number of holdings away from a
bound, less one.

\end{abstract}



\end{frontmatter}

% Body sections, one file each, in reading order (see sections/).

\section{Introduction}

\begin{figure}[t]
\centering
\includegraphics[width=0.98\columnwidth]{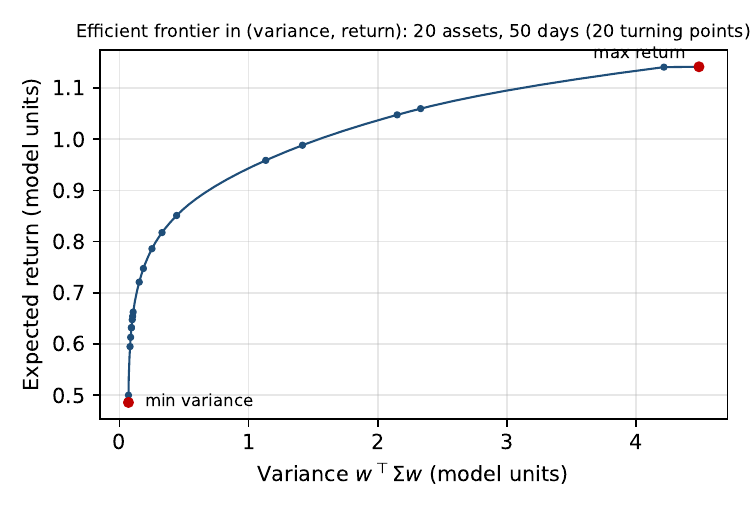}
\caption{Markowitz's efficient frontier of a $20$-asset, $50$-day factor-model problem,
computed exactly by the Critical Line Algorithm as $20$ turning points. We plot expected
return against \emph{variance} $w^\top\Sigma w$, not volatility: the variance is the
quadratic form of the program itself, which under Theorem~\ref{thm:identity} is the
LASSO's residual sum of squares. The corners are the portfolio reading of the
breakpoints that, on the regression side, form the LASSO path, after the rescaling of
Proposition~\ref{prop:bridge} that carries this long-only frontier onto the nonnegative
LASSO path and back. Called on $(X,y)$ that path stops at the maximum-Sharpe portfolio,
after $17$ of the $20$ turning points, and a shift of the response reaches the
rest.}
\label{fig:frontier}
\end{figure}

A statistician plots a LASSO path and sees kinks: the coefficients bend where the active
set changes, and the corners are the first thing on the page. A portfolio manager plots
the efficient frontier and sees a smooth arc. The corners are in that picture too (the
same active-set changes, in the same places), but they are invisible, because the
frontier meets itself with continuous slope at every one of them and only its curvature
breaks. Two conventions, one curve, and only one of them shows the
mechanism. Under $\Sigma=X^\top X$ and $\mu=X^\top y$ the two are not analogous
computations at all. They are the same computation, and this note is about what follows
from taking that seriously.

Much of modern statistical methodology is organised around solution paths: we rarely want
a single fit, but the whole family of fits as a tuning parameter varies. It is the path,
not any one point on it, that carries the information for model selection and for seeing
how a method trades fit against complexity. For an important class of these problems the
path is not something one approximates by re-fitting on a grid. It is \emph{exactly}
piecewise linear, and the whole of it can be traced, corner to corner, by a finite
algorithm that maintains nothing but an active set.

That algorithm is not native to statistics. The same procedure computes the efficient
frontier of portfolio theory, the explicit solution of a constrained control law and the
regularisation paths of statistics and machine learning, and it was worked out in
optimization and finance in the 1950s. That these are one method and not several is an
observation others have made \citep{gartner2012, mairal2012}; Remark~\ref{rem:precedents}
sets out what is prior art and what is ours. But all of them cite
\citet{markowitz1952} and none cites \citet{markowitz1956}, and the difference matters.
The 1952 paper, written while its author was a graduate student and the work his Nobel
Prize would rest on, carries almost no machinery. What it carries is the idea: that risk
and return are two objectives and not one, that the portfolios worth holding are the
undominated ones, and that they form a frontier, the first efficient frontier anyone
drew. The means of tracing it came four years later, in the \emph{Naval Research
Logistics Quarterly}. The celebrated paper posed the problem; the obscure one gave the method.

It is the 1956 paper this note is about: the \emph{procedure}, not the frontier, which is
by now common property. Writing a fast implementation of it \citep{cvxcla}, we were
struck by how closely its inner loop reproduced the LARS/LASSO homotopy statistics
arrived at four decades later \citep{efron2004}: the same active set, the same
piecewise-linear path, the same ratio test choosing the next breakpoint, with only the
names of the quantities differing. Pursued, the resemblance is an exact identity.

A reader who knows sparse portfolio selection will object at once that the connection is
made already. \citet{brodie2009} recognised the Markowitz problem under a cap on gross
exposure \emph{as} a LASSO, and \citet{gaines2018} traced the constrained-LASSO path
under general linear constraints. Both are prior art and we claim no part of either. But
\citet{brodie2009} swept a leverage penalty at fixed return rather than tracing the
risk--return frontier, so their curve meets the frontier only at the zero-penalty end,
and \citet{gaines2018} never read the constrained path as a portfolio. Neither pairs the
two \emph{procedures}, and it is the pairing that has consequences.

What is new here, then, is not that these procedures share a framework, which is also
prior art (Remark~\ref{rem:precedents}), but what follows from taking the pairing
seriously: that Markowitz's Critical Line Algorithm is an instance of the same active-set
homotopy and that its path coincides with the constrained LASSO path, under arbitrary
linear equality and inequality constraints; the map $c(\lambda)=\|\beta(\lambda)\|_1$
between their parameters, the one-to-one correspondence of breakpoints, and the single
way it degenerates; that sweeping a return tilt at fixed leverage, which is the sweep a
mandate actually asks for, is the budget sweep rescaled, and the heavier rescaling that
carries the fully invested frontier onto the nonnegative LASSO path, with the response
shifted once to reach past the maximum-Sharpe portfolio; the boundary past
which the identity stops, since $c(\lambda)$ is
computed from the data; and the degrees of freedom of the constrained fit, which in the
long-only, fully invested case is the expected number of holdings away from a bound, less
one.

Section~\ref{sec:scheme} gives the common scheme and its events,
Section~\ref{sec:identity} the identity, Section~\ref{sec:inference} what it does not
carry and the one crossing we make of that boundary, Section~\ref{sec:history} the
history, and Section~\ref{sec:utility} what the exact path shows.

\section{The common scheme}
\label{sec:scheme}

The two programs this note pairs are easy to write down. On the regression side, least
squares with an $\ell_1$ penalty and linear constraints,
$\min_\beta\tfrac12\|y-X\beta\|_2^2+\lambda\|\beta\|_1$ subject to $A\beta=b$ and
$G\beta\le h$; on the portfolio side, mean--variance selection tilted towards return,
$\min_w\tfrac12 w\trans\Sig w-\lambda\mu\trans w$ subject to a budget, a box, and
whatever else the mandate imposes. Section~\ref{sec:identity} states both precisely and
shows they trace one curve. Strip them to what they share and the common form is
immediate: both minimise a convex quadratic over a polyhedron, both carry an $\ell_1$
term, and in both a single scalar moves. Written once:
\begin{equation}
\label{eq:pqp}
\begin{aligned}
&x(\lambda) = \arg\min_{x\in\R^n}\ \tfrac12\, x\trans H x - c(\lambda)\trans x
   + \tau(\lambda)\,\|x\|_1 \\
&\quad\text{s.t.}\ \ A x = b,\ \ G x \le h,\ \ \lb \le x \le \ub ,
\end{aligned}
\end{equation}
with $H$, $A$ and $G$ fixed, the constraint data $b,h,\lb,\ub$ fixed for now, and the two
moving parts affine in the parameter: $c(\lambda)=q+\lambda d$ splits into a piece that
stays and a piece that moves, and the penalty weight
$\tau(\lambda)=\tau_0+\lambda\tau_1$ is a datum like the rest, required only to stay
nonnegative on the interval traced. The constrained LASSO takes $\tau(\lambda)=\lambda$
with $q=X\trans y$ and $d=0$; the frontier takes $\tau\equiv0$ and the tilt $d=\mu$.
Nothing else distinguishes them here, and nothing below depends on where in
\eqref{eq:pqp} the parameter sits, only that it enters affinely; that alone makes the
right-hand side of the bordered system \eqref{eq:border} affine, and hence the segment.

A penalty is not a constraint in disguise, which is why \eqref{eq:pqp} carries both. The
ball $\|x\|_1\le c$ is a polyhedron with $2^n$ facets and is not among the rows of $A$,
$G$ or the bounds; a template meant to contain both forms must carry both. What
Section~\ref{sec:identity} shows is the different and sharper fact that the two
\emph{paths} coincide.

\paragraph{Letting the constraint data move} Everything below holds verbatim if
$b,h,\lb,\ub$ are themselves affine in $\lambda$, and two members need that. The
rescaled frontier of Proposition~\ref{prop:bridge} moves its budget row and its box
together; the support-vector-machine path of \citet{hastie2004svm} carries the cost $C$
as the upper bound of a box $0\le\alpha\le C$ and nowhere else. The cost is nil: blocked
coordinates $x_\B$ are then affine in $\lambda$ rather than constant, so the right-hand
side of \eqref{eq:border} is affine either way and Lemma~\ref{lem:affine} is untouched.
A moving box is to $\mathrm{P}_1$ exactly what the $\ell_1$ subgradient's travelling
threshold is to $\mathrm{D}_1$ (Section~\ref{sec:events}): the same ratio test, against
a boundary in motion. In that generality the rule (fixed matrices, every other datum
affine) goes back to \citet{ritter1962}, and \citet{gartner2012} state it in exactly
this form. The generalized lasso replaces $\|x\|_1$ by $\|Dx\|_1$ and runs the same
argument with $D$ in place of the identity, the rows pinned at $(Dx)_i=0$ joining $C$
rather than collapsing onto blocked coordinates.

\paragraph{The standing assumption} We require only that $H\succeq0$ and that, at every
active set the path visits, the active rows $C_\F$ on the free block (defined below)
have full row rank and the \emph{reduced} Hessian $Z\trans H_{\F\F}Z$ is positive
definite, where the columns of $Z$ span the null space of $C_\F$. With no active rows
$Z$ is the identity and the condition reads $H_{\F\F}\succ0$. For $H\succeq0$ the two
conditions together are exactly what makes the bordered matrix of \eqref{eq:border}
nonsingular (Appendix~\ref{sec:ratios}). This is the textbook condition for an
equality-constrained quadratic program \citep[Lemma~16.1]{nocedal2006}, and it is the
one Markowitz already used (Section~\ref{sec:history}). It is weaker than the textbook
$H\succ0$, and deliberately so: it is exactly what the regression instance needs. There
$H=X\trans X$ is singular whenever the design has more columns than rows (more
predictors than observations, the very regime the LASSO is built for), yet
the free block $X_\F\trans X_\F$ is positive definite as soon as the active columns are
linearly independent, which holds in general position
(Remark~\ref{rem:unique}). Constraints take it further. $X_\F Z$ can have full column
rank where $X_\F$ does not, so under a row such as $\mathbf 1\trans\beta=0$ the free set
may hold more coordinates than $X$ has rows and the condition still hold. The finance instance has $H=\Sig\succ0$ whenever the covariance is estimated from
more observations than assets, or carries an idiosyncratic floor, which is the usual
case and makes the condition automatic; a rank-deficient $\Sig$ puts finance in the same
position as a regression with more predictors than observations, and the reduced-Hessian condition is then what carries
both. Under it the bordered solve below is well-posed, which is
all the homotopy ever touches; a globally positive-definite $H$ is never needed. It
is, however, an assumption and not a formality, and Remark~\ref{rem:rank} says where it
bites.

\begin{remark}[Rank deficiency, and the limits of the standing assumption]
\label{rem:rank}
The assumption concentrates the whole difficulty of the rank-deficient case into one
line, and we flag rather than dissolve it. When $Z\trans H_{\F\F}Z$ is singular at some
visited active set the bordered system \eqref{eq:border} no longer determines the segment, and
what the path even \emph{is} becomes a question of which minimiser one selects. The
existing treatments answer that question in incompatible ways, and none of them is the
common scheme run with weaker hypotheses; Section~\ref{sec:conclusion} sets out how they
divide and what it leaves open.
\end{remark}

\begin{table}[t]
\centering
\footnotesize
\caption{Notation. Roman capitals are fixed data matrices; calligraphic capitals are
index sets that change with the active set. Subscripting by a set restricts to those
rows, columns or coordinates.}
\label{tab:notation}
\begin{tabular}{@{}>{\raggedright\arraybackslash}p{0.18\columnwidth}>{\raggedright\arraybackslash}p{0.72\columnwidth}@{}}
\toprule
Symbol & Meaning \\
\midrule
$H$ & the quadratic form; $\Sig$ in finance, $X\trans X$ in regression \\[2pt]
$A$, $G$ & the fixed equality and inequality matrices of \eqref{eq:pqp} \\[2pt]
$\tau(\lambda)$ & the penalty weight, affine in $\lambda$; $\tau\equiv0$ for the
constraint-driven members \\[2pt]
$\B$ & blocked coordinates, pinned at a bound or at zero \\[2pt]
$\F=\B^{c}$ & the free coordinates; the \emph{support} when the penalty pins the rest \\[2pt]
$\Sset$ & the active inequality rows of $G$ \\[2pt]
$C_\F$, $r(\lambda)$ & the active rows on the free block, and their right-hand side \\[2pt]
$s_\F$ & the signs of the free coordinates \\[2pt]
$\zeta$, $\rho$ & multipliers for the active rows and for the blocked coordinates \\
\bottomrule
\end{tabular}
\end{table}

\begin{lemma}[Affine segments and the bordered solve]
\label{lem:affine}
On any maximal interval of $\lambda$ on which the active set is constant, the optimiser
of \eqref{eq:pqp} is affine,
\begin{equation}
\label{eq:affine}
x(\lambda) = \alpha + \lambda\,\delta ,
\end{equation}
and $(\alpha,\delta)$ come from a single bordered linear solve \eqref{eq:border} whose
matrix is constant on the interval. An active inequality row enters that solve
identically to an equality row.
\end{lemma}

The lemma is classical. \citet{ritter1962}, translated as \citet{ritter1967}, partitions
the parameter interval into finitely many subintervals on each of which the optimum
solves a parameter-dependent system of equations, linear for a quadratic objective under
linear constraints; the active set changes at each interval end by adding, deleting, or
exchanging a constraint. We set the system up rather than reprove it, because the
notation is the one the rest of the note uses. At a given
$\lambda$ the active set names the equalities $A$, always active; the active inequality
rows $\Sset$, those held at $g_i\trans x = h_i$, the remainder carrying a zero multiplier
by complementary slackness; and the blocked variables $\B$, each pinned either at a bound
or, when $\tau\not\equiv0$, at zero by the penalty, with free set $\F=\B^{c}$; under a
penalty every free coordinate is therefore nonzero and carries a fixed sign
$s_i=\operatorname{sign}(x_i)$. Stack the active linear constraints, with their
right-hand side, as
\[
C=\begin{bmatrix} A \\[2pt] G_{\Sset} \end{bmatrix},
\qquad
r(\lambda)=\begin{bmatrix} b(\lambda) \\[2pt] h_{\Sset}(\lambda) \end{bmatrix},
\]
writing $r(\lambda)$ for the general case, constant in $\lambda$ unless the constraint
data move. Since \eqref{eq:pqp} is convex with linear
constraints and a convex penalty the KKT
conditions are necessary and sufficient, and reading the stationarity and active-row
feasibility conditions on the free block, with $x_\B$ held at its bounds, collects the
unknowns $(x_\F,\zeta)$ into
\begin{equation}
\label{eq:border}
\begin{bmatrix} H_{\F\F} & C_\F\trans \\[2pt] C_\F & 0 \end{bmatrix}
\begin{bmatrix} x_\F \\[2pt] \zeta \end{bmatrix}
=
\begin{bmatrix} c_\F(\lambda) - \tau(\lambda)\,s_\F
   - H_{\F\B}\,x_\B \\[2pt] r(\lambda) - C_\B\,x_\B \end{bmatrix},
\end{equation}
where $\zeta$ carries the equality and active-inequality multipliers. The coefficient
matrix is the saddle-point matrix of the active constraints, nonsingular under the
standing assumption, and it depends only on $H$ and the active normals, so it is
constant on the interval while the right-hand side is affine in $\lambda$. That is the
lemma. Appendix~\ref{sec:ratios} records the sign convention, what the bordered block
absorbs, and the elimination that solves \eqref{eq:border} in practice.

One step of that reading is the penalty's, and we give it because it is the one that lets
one template carry both face mechanisms. The subdifferential of $\tau(\lambda)\|x\|_1$
at $x$ is $\tau(\lambda) g$ with $g_i=s_i$ where $x_i\neq0$ and $g_i\in[-1,1]$ where
$x_i=0$. On the free block it is data, affine in $\lambda$, and joins the right-hand side
as $\tau(\lambda)s_\F$. Where $x_i=0$ it is not data, and the coordinate is blocked
exactly as a bound blocks one, its multiplier $\rho_i$ confined to
$[-\tau(\lambda),\tau(\lambda)]$ rather than to a one-sided cone. So a coordinate the
penalty pins is a blocked variable like any other, and the price is one dual-feasibility
condition, $|\rho_i|\le\tau(\lambda)$, carried alongside
the usual $\zeta_{\Sset}\ge0$. That condition is the moving threshold of
Section~\ref{sec:events}, and the reason the $\ell_1$ event fires at $|\rho_i|=\lambda$
rather than at a fixed bound.

\begin{remark}[The LASSO, read off \eqref{eq:border}]
\label{rem:lassocheck}
Take $H=X\trans X$, $q=X\trans y$, $d=0$, $\tau(\lambda)=\lambda$, no $A$, $G$ or
bounds. Then $\B$ is
the set of zero coefficients and $\F$ the support with signs $s_\F$, $C$ is empty, and
\eqref{eq:border} collapses to
\[
X_\F\trans X_\F\, x_\F = X_\F\trans y - \lambda\, s_\F ,
\]
the familiar LASSO segment equation, with offset
$\alpha_\F=(X_\F\trans X_\F)^{-1}X_\F\trans y$ and slope
$\delta_\F=-(X_\F\trans X_\F)^{-1}s_\F$. The $(\F,\B)$ partition carries the whole of
the penalty's bookkeeping: the two procedures are instances of the same active-set KKT
mechanism, the $\ell_1$ penalty and the box producing analogous primal and dual events,
and $\tau$ is what selects between them.
\end{remark}

That penalty-only specialisation is the statistical one, and it is the squared-loss,
$\ell_1$ case of the piecewise-linear-path characterisation of \citet{rosset2007}, who
generalise in a direction we do not pursue, to non-quadratic losses. We carry the
constrained form because the statistical literature does not: no single penalty
$\lambda J$ can express general equalities, inequalities and box constraints at once,
since $\lambda$ scales the penalty whereas the constraints of \eqref{eq:pqp} are fixed.
That, with the converse noted above, is what admits the constraint-driven members of the
family, the Critical Line Algorithm and the constrained LASSO among them.

The intervals partition the $\lambda$-axis at \emph{breakpoints} where the active set
must change. Continuity of $x(\lambda)$ across breakpoints (the free block is
uniquely determined on each segment, and adjacent segments agree at the breakpoint)
lets a \emph{homotopy} recover the whole path: from a vertex with known
active set, follow \eqref{eq:affine} to the first event, update the active set by that
one event, repeat. Each step is one solve against the active block of $H$ bordered by
the active normals. We call \eqref{eq:pqp}, \eqref{eq:affine}, and this homotopy the
\emph{common scheme}. Lemma~\ref{lem:affine} is only half of it, and the less
interesting half: it says what a segment is, not how one segment gives way to the next.
That second half, the evolution of the active set, is where each field's paper does its
technical work, and we give it next.

\subsection{How the active set evolves}
\label{sec:events}

A representation is not an algorithm. Lemma~\ref{lem:affine} supplies a \emph{master
representation}, the bordered system whose solution is the segment; what turns it into a
\emph{master algorithm} is the catalogue of events that can end a segment, each with the
step to it, and the update each event makes to $(\F,\B,\Sset)$. This is the part that
the instances present in mutually unrecognisable notation, as entering and leaving
variables, as assets reaching and releasing their bounds, as points crossing the margin,
or as a critical-region facet, and it is the part in which their equivalence is least
obvious. We write it once.

Fix the current breakpoint $\lambda_k$ and a direction of travel; write
$\lambda=\lambda_k+\sigma t$ with $\sigma\in\{+1,-1\}$ and call $t\ge0$ the step. Two
solves of \eqref{eq:border} against one factorisation give the segment
$x_\F(\lambda)=\alpha_\F+\lambda\delta_\F$ and the multipliers $\zeta(\lambda)$, and the
bound multipliers $\rho(\lambda)$ follow from stationarity; all are affine in $\lambda$,
so every quantity below moves at a constant rate. Four things, and only four, can end the
segment: a primal quantity can reach a bound it must respect, or a dual quantity can
reach the edge of the cone it must lie in. Each gives a ratio, the step is the smallest
of them, and the event attaining it fires.

\emph{$\mathrm{P}_1$: a free variable reaches a bound.} For $i\in\F$ the coordinate $x_i$
must stay in $[\lb_i,\ub_i]$; the event moves $i$ from $\F$ to $\B$, pinned where it
landed.

\emph{$\mathrm{P}_2$: an inactive inequality becomes active.} For $j\notin\Sset$ the
slack $h_j-g_j\trans x$ must stay non-negative; the event adjoins $j$ to $\Sset$, adding
a row to $C$.

\emph{$\mathrm{D}_1$: a blocked variable releases.} For $i\in\B$ dual feasibility confines
$\rho_i(\lambda)$ to an interval whose endpoints are themselves affine in $\lambda$; the
event is the first zero of either margin and moves $i$ from $\B$ to $\F$.

\emph{$\mathrm{D}_2$: an active inequality releases.} For $j\in\Sset$ the multiplier must
satisfy $\zeta_j\ge0$; the event deletes $j$ from $\Sset$. Equality rows of $A$ never
generate it, their multipliers being unsigned.

Writing each quantity as its value at $\lambda_k$ plus $\sigma t$ times its rate turns
each into a ratio, and the step $t^\star$ is the smallest of the four;
Appendix~\ref{sec:ratios} collects the formulas, where $\mathrm{P}_1$ and $\mathrm{D}_1$
turn out to be the same test, one against the box and one against the cone.

What is worth dwelling on is $\mathrm{D}_1$, because it is where the
two kinds of active set meet. For a genuine box the confining interval is $[0,\infty)$
at an upper bound and $(-\infty,0]$ at a lower one, so $\dot\lb^\rho=\dot\ub^\rho=0$ and
the event is simply $\rho_i=0$, a fixed threshold. For the $\ell_1$ penalty the interval
is $[-\tau(\lambda),\tau(\lambda)]$, so $\dot\lb^\rho=-\tau_1$ and
$\dot\ub^\rho=+\tau_1$, and the event is $|\rho_i|=\tau(\lambda)$, an affine function
meeting a threshold that travels with $\lambda$. That is the condition $|\rho_i|\le\tau(\lambda)$
on a coordinate the penalty pins, derived above and inherited from $g_i\in[-1,1]$. That single generalisation, a feasibility boundary
in motion, is the whole of what separates an \emph{induced} active set from an
\emph{imposed} one at the level of the algorithm.

The walk continues until no event is reachable. It is the simplex method's ratio test and
inherits the simplex method's degeneracy hazard, ties at $t^\star=0$ and cycling, and
borrows its usual remedy, a lowest-index (Bland) tie-break. The tie-break makes the walk
deterministic; the simplex method's proof that Bland's rule terminates is not shown to
carry over, and whether the walk is finite in every degenerate case is among the
questions left open in Section~\ref{sec:conclusion}. A tie-break cannot settle it alone:
\citet{hobson2026classo} show that a homotopy can fail to terminate even where only one event is
available, when the multipliers are not unique and a different certificate is chosen at
each step, and they prove termination by fixing that choice. Because exactly one index changes status per step, the factorisation of the
bordered matrix is updated, not rebuilt.

\begin{table}[t]
\centering
\footnotesize
\caption{The four events of Section~\ref{sec:events} as each literature names them. The
algorithms differ in which families are populated, not in the families.}
\label{tab:events}
\begin{tabular}{@{}>{\raggedright\arraybackslash}p{0.14\linewidth}>{\raggedright\arraybackslash}p{0.36\linewidth}>{\raggedright\arraybackslash}p{0.38\linewidth}@{}}
\toprule
Event & LASSO\,/\,LARS & Critical Line Algorithm \\
\midrule
$\mathrm{P}_1$ & a coefficient crosses zero (the LASSO ``leave'' move) & an asset reaches
a bound and leaves the free set \\[2pt]
$\mathrm{P}_2$ & a linear inequality on $\beta$ becomes tight & a group or turnover
constraint becomes tight \\[2pt]
$\mathrm{D}_1$ & an inactive coordinate's correlation reaches $\lambda$ (the ``enter''
move) & an asset's reduced cost changes sign and it re-enters \\[2pt]
$\mathrm{D}_2$ & an active inequality's multiplier hits zero & a constraint stops binding
at a turning point \\
\bottomrule
\end{tabular}
\end{table}

Table~\ref{tab:events} reads the catalogue back into two of the dialects. The instances
differ in which families they populate, not in the families themselves. Forward LARS has
$\mathrm{D}_1$ only, which is exactly what makes its active set monotone; the LASSO
modification restores $\mathrm{P}_1$ and with it the leave move; the Critical Line
Algorithm populates all four against fixed rather than moving feasibility intervals,
\citet{markowitz1955} carrying general inequalities from the outset and giving every
crossing explicitly; the constrained LASSO of \citet{gaines2018} is that catalogue
rediscovered in the statistical setting;
and the support-vector-machine path of
\citet{hastie2004svm} is $\mathrm{P}_1$ and $\mathrm{D}_1$ with the margin supplying the
moving threshold. The catalogue is a statement about a scalar step; Remark~\ref{rem:mpc} says how much of
it survives when the parameter is a vector, which is less than the segment
representation does. Appendix~\ref{sec:ratios} writes the loop out.

\noindent We should be careful about what this is. It is not the most general
path-following algorithm for a parametric quadratic program: \citet{gartner2009} give one
that is more general in the direction that matters most, since it handles a singular
quadratic form without perturbing the problem, which Remark~\ref{rem:rank} is precisely
unable to do. What the box gives is the four events written once, in a notation close
enough to each field's own that the specialisations are checkable by inspection, which is
what Table~\ref{tab:events} and Remark~\ref{rem:larsstep} then do.

\begin{remark}[Where the analogy stops: explicit model predictive control]
\label{rem:mpc}
The constrained linear-quadratic regulator solved as a function of the state
\citep{bemporad2002, tondel2003} is the multidimensional cousin: the parameter is a
\emph{vector}, the state, so the interval of breakpoints becomes a polyhedral partition
into critical regions. It is not that the parameter enters in the wrong place. The state
enters the constraint right-hand side, $b(\theta)$ and $h(\theta)$, which \eqref{eq:pqp}
admits, and along a fixed ray through parameter space the mpQP is an instance of the
common scheme, data and all, ratio test included. What does not carry over is the
\emph{algorithm}. Lemma~\ref{lem:affine} survives, which is why the explicit control law
is piecewise affine; but exploring a polyhedral partition is not a walk along a line, and
its difficulties (degenerate critical regions, redundant facets, region merging, and
the combinatorics of which facet to cross next) have no counterpart in the scalar case
and are the subject matter of that literature. Region traversal is not the ratio test
with a direction substituted for a sign, and we claim correspondingly less for this
member than for the others.
\end{remark}

\paragraph{$H$ is an operator, not a matrix}
Inspecting the step of Lemma~\ref{lem:affine}, the quadratic form enters only
through three linear actions: a matrix--vector product $v\mapsto Hv$, a solve $Hv=r$
restricted to the active block, and the cross product coupling active and inactive
coordinates. Anything supplying those three can play the role of $H$, which is never
needed as a stored $n\times n$ matrix. A factor covariance, a kernel and a sparse
design each supply the three cheaply, which is what turns one algorithm into a method for
many problems.

\section{The exact identity}
\label{sec:identity}

We now show that the two most dissimilar members of the family trace, under one
substitution, the same curve, and that the identity is robust to general constraints.
The claim is about the paths and not about the forms: \eqref{eq:lasso} is an instance of
the common scheme, while \eqref{eq:mark} reaches the same curve through the theorem
rather than by being one. Figure~\ref{fig:correspondence} is the whole of it in one
picture.

\begin{figure*}[t]
\centering
\includegraphics[width=0.96\textwidth]{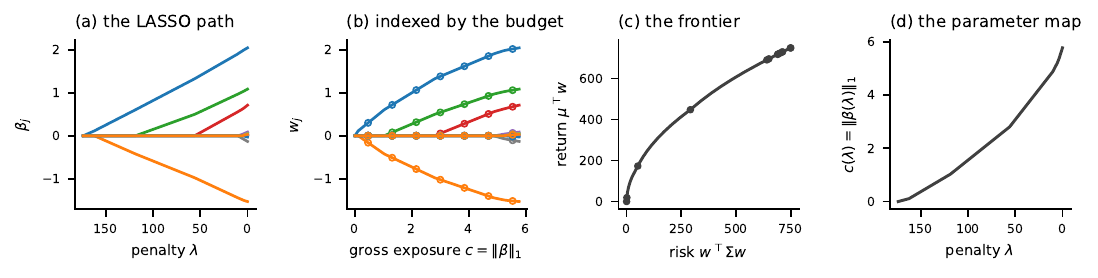}
\caption{One curve, four readings. All four panels come from a single $(X,y)$, with
$\Sig=X\trans X$ and $\mu=X\trans y$. \emph{(a)} the LASSO coefficient path against the
penalty. \emph{(b)} the same coefficients against the gross exposure
$c=\|\beta\|_1$, which is \eqref{eq:mark}. \emph{(c)} the same points in the
(risk, return) plane, which is the efficient frontier traced corner to corner.
\emph{(d)} the map $c(\lambda)$ of Theorem~\ref{thm:identity} between the two
parametrisations. The open markers in \emph{(b)} are the $\ell_1$-capped least-squares
problem solved independently as a quadratic program at seven caps, agreeing with the
homotopy to solver precision. Panels \emph{(a)} and \emph{(b)} differ only by that map, which is why
they look alike: Proposition~\ref{prop:reg} shows it is strictly decreasing here, so the
two are reparametrisations of one curve and not two curves that resemble each other.}
\label{fig:correspondence}
\end{figure*}

\subsection{The efficient frontier}
\label{sec:frontier-primer}

The efficient frontier is an object the statistician already knows under another name. An
investor spreads weights $w$ over $n$ assets and wants high expected return $\mu\trans w$
and low risk $w\trans\Sig w$ at once. No portfolio does both, so the object of study is
the set of undominated ones, and it is traced by scalarisation: minimise
$\tfrac12 w\trans\Sig w-\lambda\,\mu\trans w$ subject to the budget
$\mathbf 1\trans w=1$ and whatever else the mandate imposes. That is \eqref{eq:pqp} with
$H=\Sig$, $d=\mu$ and no penalty ($q=0$, $\tau\equiv0$), so the frontier is an instance
of the common scheme before any identity is proved. Sweeping $\lambda$ from $0$ to
$\infty$ runs from minimum variance to maximum return and traces the frontier exactly, as
parabolic arcs joined at \emph{corner portfolios}, which are its breakpoints. We draw
the frontier in the (variance, return) plane throughout. The variance is the quadratic
form of the program itself, and under Theorem~\ref{thm:identity} the LASSO's residual sum
of squares, so the arcs are genuinely parabolic and the geometric statements below are
statements about the objective rather than about its square root; the volatility
convention is a monotone relabelling of one axis and changes none of them.

Those arcs meet with continuous slope: $\mathrm{d}V/\mathrm{d}E$ is proportional to the
tilt $\lambda$, which moves continuously along the path, so only the curvature breaks at
a corner, and the square root preserves that wherever $V>0$. The two fields have been
drawing the same corners for seventy years, and only one of them draws them where they
show.

The frontier is therefore a regularisation path in all but name, differing only in which
objectives are traded: return against variance here, fit against sparsity for the LASSO.
One caution on the parameter. This $\lambda$ is a \emph{return tilt} and not a risk
aversion, and Section~\ref{sec:bridge} shows that reaching the identity from the frontier
costs a rescaling because of it.

\subsection{Two programs, one curve}

\paragraph{The two programs}
For weights $w\in\R^n$, expected returns $\mu$, and covariance $\Sig\succeq 0$
(positive \emph{semi}definite: under $\Sig=X\trans X$ the finance instance's
$\Sig\succ0$ would exclude $n>m$, the regime the LASSO is built for, and the standing
assumption of Section~\ref{sec:scheme} is what is actually needed), the
gross-exposure-constrained mean--variance program is
\begin{equation}
\label{eq:mark}
\begin{aligned}
\text{M}_c:\quad \min_{w}\ &\tfrac12\, w\trans \Sig\, w - \mu\trans w \\
\text{s.t.}\quad &\|w\|_1 \le c,\ \ A w = b,\ \ G w \le h .
\end{aligned}
\end{equation}
For a response $y\in\R^m$ and design $X\in\R^{m\times n}$, so that $m$ counts
observations and $n$ counts predictors (assets, on the portfolio side), the
constrained LASSO is
\begin{equation}
\label{eq:lasso}
\begin{aligned}
\text{L}_\lambda:\quad \min_{\beta}\ &\tfrac12\,\|y - X\beta\|_2^2 + \lambda\,\|\beta\|_1 \\
\text{s.t.}\quad &A\beta = b,\ \ G\beta \le h .
\end{aligned}
\end{equation}
The two programs keep their fields' conventional symbols, $w$ for portfolio weights and
$\beta$ for regression coefficients; under the substitution of Theorem~\ref{thm:identity}
these are one and the same vector, and below we identify them freely. The substitution
runs both ways: from the regression side $(X,y)$ are given and $\Sig,\mu$ read off them,
while from the portfolio side any $\Sig\succeq0$ factors as $X\trans X$ and $\mu=X\trans
y$ is solvable exactly when $\mu\in\operatorname{range}(\Sig)$, which is to require
that the mean--variance problem have a finite answer, since a $\mu$ outside that range
promises return along a direction the risk model believes carries no risk.

\paragraph{The leverage axis} Programme \eqref{eq:mark} caps the gross exposure
$\|w\|_1$, which the frontier of Section~\ref{sec:frontier-primer} never mentions.
Adding it makes the efficient set a \emph{surface} in (return, risk, leverage), of which
a frontier is a one-dimensional cut; which cut one takes is a choice, and different
literatures have taken different ones (Remark~\ref{rem:precedents}). Less of a choice
than it appears: where the remaining constraints are homogeneous, that surface is the
cone over a single curve, and the cuts differ by a rescaling rather than in their
geometry (Corollary~\ref{cor:tilt}). A long-only, fully
invested mandate pins $\|w\|_1=\mathbf 1\trans w=1$, so the cap does no work and the
surface degenerates to the familiar curve. Long--short books are what make the axis live,
gross exposure varying there and costing something, and they are what \eqref{eq:mark} is
for.

\begin{theorem}[CLA--LASSO identity under general constraints]
\label{thm:identity}
Suppose $\Sig = X\trans X$ and $\mu = X\trans y$, and suppose the columns of $X$ are in
general position and the linear-independence constraint qualification holds at every
active set the path visits. Then $\mathrm{L}_\lambda$ has a unique minimiser
$\beta(\lambda)$ at every $\lambda$, so that $\lambda\mapsto\beta(\lambda)$ is a
\emph{single-valued} path and the homotopy of Section~\ref{sec:scheme} is well defined.
Everything below is a statement about that path. None of it extends to the
rank-deficient case, where the minimiser may be a set rather than a point and there is no
curve to speak of (Remark~\ref{rem:rank}).

Put $c(\lambda)=\|\beta(\lambda)\|_1$, which is nonincreasing, and continuous by
Proposition~\ref{prop:reg}.
Then $\beta(\lambda)$ solves $\mathrm{M}_{c(\lambda)}$ for every $\lambda\ge0$, and on
every interval of $\lambda$ on which $c$ is \emph{strictly} decreasing the two programs
have the same solution set (so $\mathrm{M}_c$ is single-valued there as well) and
their paths are the \emph{same} piecewise-linear curve in $\R^n$, traced by the same
active-set homotopy and sharing the same breakpoints as points of $\R^n$.
Proposition~\ref{prop:flat} treats the remaining case, in which $\mathrm{M}_c$ acquires
a continuum of minimisers and the two curves genuinely part.
\end{theorem}

\begin{proof}
Completing the square,
\[
\tfrac12\|y-Xw\|_2^2 = \tfrac12\, w\trans (X\trans X)\, w - (X\trans y)\trans w
+ \tfrac12\|y\|_2^2 .
\]
With $\Sig=X\trans X$ and $\mu=X\trans y$ the
objective of $\mathrm{M}_c$ equals $\tfrac12\|y-Xw\|_2^2$ up to the additive constant
$\tfrac12\|y\|_2^2$, which does not affect the minimiser. Hence, under the same linear
constraints $Aw=b,\ Gw\le h$, $\mathrm{M}_c$ is exactly the constrained least-squares
problem $\min_w \tfrac12\|y-Xw\|_2^2$ subject to $\|w\|_1\le c$ and those constraints.
Its Lagrangian relaxation of the single constraint $\|w\|_1\le c$ is
$\min_w \tfrac12\|y-Xw\|_2^2 + \lambda\|w\|_1$ subject to $Aw=b,\ Gw\le h$, i.e.\
$\mathrm{L}_\lambda$. Both directions of the pointwise correspondence are then standard,
and we claim nothing for them. One needs no duality at all: $\beta(\lambda)$ is feasible
for $\mathrm{M}_{c(\lambda)}$ by construction, and any strictly better point for
$\mathrm{M}_{c(\lambda)}$ would, having $\ell_1$ norm at most $c(\lambda)$, be strictly
better for $\mathrm{L}_\lambda$ too. The other is Lagrangian duality: for a cap $c$
admitting a point of $\{A w=b,\ Gw\le h\}$ with $\|w\|_1<c$, Slater's condition holds
(the remaining rows are affine), so there is no duality gap and an optimal
multiplier $\lambda^\star\ge0$ exists at which the minimisers of $\mathrm{L}_{\lambda^
\star}$ that are $\mathrm{M}_c$-feasible are exactly the minimisers of $\mathrm{M}_c$.
Slater fails only at the smallest feasible cap, where the ball is tight throughout and
the two problems agree trivially. What strict monotonicity then adds is
\emph{injectivity}: the correspondence matches breakpoints one-to-one rather than merely
covering them. Both
objectives are quadratic and both feasible sets polyhedral, so each path is piecewise
linear. For $\mathrm{L}_\lambda$ that is Lemma~\ref{lem:affine} directly. For
$\mathrm{M}_c$ it is Lemma~\ref{lem:affine} one interval at a time: $\mathrm{M}_c$ is
not globally of the form \eqref{eq:pqp}, since the active facet of $\|w\|_1\le c$ has
normal the current sign pattern $s$, which \eqref{eq:pqp} holds fixed; but on an interval
where the active set is constant so is $s$, there $\|w\|_1=s\trans w$, and
$\mathrm{M}_c$ agrees with the instance of \eqref{eq:pqp} whose $G$ carries the extra
row $s\trans w\le c$. The intervals are finite in number. Sharing a solution set at every parameter value, they
are the same set of points with the same corners, visited in the same order.
\end{proof}

\begin{corollary}[One-to-one breakpoints]
\label{cor:bijection}
If in addition the support's sign vector $s_\F$ fails to lie in the span of the active
constraint normals on every segment the path visits (automatic in the unconstrained
case, by Proposition~\ref{prop:reg}), then $c$ is
strictly decreasing throughout, $\lambda\mapsto c(\lambda)$ is a bijection onto its
range, and the two breakpoint sequences match one-to-one.
\end{corollary}

\subsection{The parameter map, and where it fails}

Both curves are now in hand, and both are indexed by parameters of their own. What
remains is the map between those parameters, $c(\lambda)$ of Theorem~\ref{thm:identity},
which can be pinned down exactly. Doing so sharpens the theorem and isolates the only way
it degenerates.

\begin{proposition}[Regularity of the $c$--$\lambda$ correspondence]
\label{prop:reg}
Let $\beta(\lambda)$ be the constrained LASSO path and
$c(\lambda)=\|\beta(\lambda)\|_1$ as in Theorem~\ref{thm:identity}. Fix a segment, with
free block $\F$ carrying signs $s_\F$, and let $\delta_\F$ be the free block of the
solution of \eqref{eq:border} with right-hand side $(-s_\F;\,0)$, the slope of
Lemma~\ref{lem:affine} for the $\ell_1$ segment. Then $c$ is continuous and piecewise
linear, with
\[
c'(\lambda)=s_\F\trans\delta_\F ,
\]
the coordinates held at a bound contributing a constant to $\|\beta(\lambda)\|_1$ and so
not entering the slope. In the penalty-only case the blocked coordinates are the zero
ones, $\F$ is the support, and $\delta_\F=-(X_\F\trans X_\F)^{-1}s_\F$, so
$c'(\lambda)=-s_\F\trans(X_\F\trans X_\F)^{-1}s_\F<0$ on every segment with nonempty
support and $c$ is a strictly decreasing, piecewise-linear bijection from
$(0,\lambda_{\max}]$ onto $[0,\|\beta(0)\|_1)$, where $\lambda_{\max}$ is the smallest
$\lambda$ at which the support is empty. Under constraints $c'(\lambda)\le0$
always, and the bijection survives whenever the reduced Hessian on the active constraint
tangent space is positive definite and $s_\F$ is not annihilated by it.
\end{proposition}

\begin{proof}
On a segment $\F$ and $s_\F$ are fixed, so $\|\beta(\lambda)\|_1$ equals
$s_\F\trans(\alpha_\F+\lambda\delta_\F)$ plus the constant contributed by the blocked
coordinates, which is affine with the stated slope, and continuity across breakpoints
follows from continuity of $\beta(\lambda)$. Unconstrained, the segment stationarity
gives
\[
X_\F\trans X_\F\,\beta_\F = X_\F\trans y-\lambda s_\F ,
\qquad
\delta_\F=-(X_\F\trans X_\F)^{-1}s_\F ,
\]
and $s_\F\trans(X_\F\trans X_\F)^{-1}s_\F>0$ since $X_\F\trans X_\F\succ0$ and
$s_\F\neq0$ on a nonempty support; a strictly
negative derivative on every segment makes $c$ a strictly decreasing bijection onto its
range. In the constrained case the second block row of \eqref{eq:border} with right-hand
side $(g;0)$ puts the free block in the null space of $C_\F$, so it is $Zu$ with the
columns of $Z$ a basis for that null space, and multiplying the first block row by
$Z\trans$ removes $\zeta$ and leaves $Z\trans H_{\F\F}Z\,u=Z\trans g$. The free block of
the response is therefore $Z(Z\trans H_{\F\F}Z)^{-1}Z\trans g$, which needs only the
standing assumption and not $H_{\F\F}\succ0$. Taking $g=-s_\F$ gives
$s_\F\trans\delta_\F=-\lVert(Z\trans H_{\F\F}Z)^{-1/2}Z\trans
s_\F\rVert_2^2\le0$, with equality only when $Z\trans s_\F=0$, that is when
$s_\F$ lies in the span of the active normals; otherwise it is strictly negative and
the bijection follows as before.
\end{proof}

Two riders on the statement. The blocked coordinates contribute a constant to
$\|\beta(\lambda)\|_1$ because \eqref{eq:lasso} carries a fixed box; where the bounds
travel with the parameter, as in \eqref{eq:nnls}, that contribution is affine instead
and its rate joins the slope. And under constraints that exclude $\beta=0$, such as
$A\beta=b$ with $b\neq0$, the support need never empty, so no $\lambda_{\max}$ exists in
the unconstrained sense and the displayed bijection is the unconstrained statement;
Proposition~\ref{prop:ends} supplies the replacement. When $s_\F$ is annihilated $c$ is
constant on the segment, and if $\beta$ still moves there, what is lost is more than the
bijection.

\begin{proposition}[Flat segments break the identity of curves]
\label{prop:flat}
Suppose $c\equiv c_0$ on $[\lambda_1,\lambda_2]$ with
$\beta(\lambda_1)\neq\beta(\lambda_2)$. Then $\mathrm{M}_{c_0}$ has a continuum of
minimisers, containing the whole segment $\{\beta(\lambda):\lambda_1\le\lambda\le
\lambda_2\}$. Consequently a single-valued tracer of $\mathrm{M}_c$ returns one point
of that segment, and the path it traces omits the rest: the two paths are \emph{not} the
same set of points.
\end{proposition}

\begin{proof}
Write $R(\beta)=\tfrac12\|y-X\beta\|_2^2$ and $\beta_i=\beta(\lambda_i)$. Optimality of
$\beta_1$ for $\mathrm{L}_{\lambda_1}$ gives
$R(\beta_1)+\lambda_1c_0\le R(\beta_2)+\lambda_1c_0$, hence $R(\beta_1)\le R(\beta_2)$;
the same argument at $\lambda_2$ gives the reverse, so $R(\beta_1)=R(\beta_2)$. Both are
feasible for $\mathrm{M}_{c_0}$, both attain the same objective, and neither can be
beaten there, since a strictly better point would have $\ell_1$ norm at most $c_0$ and
so would beat $\beta_1$ in $\mathrm{L}_{\lambda_1}$. The same holds at every intermediate
$\lambda$, and the segment is connected.
\end{proof}

So the flat case is not an axis relabelled non-injectively; it is a genuine divergence
of the two curves, and it is why Theorem~\ref{thm:identity} restricts to intervals of
strict decrease. Where it can happen is exactly where this note's contribution lies. If
$X\trans X\succ0$ the objective is strictly convex, $\mathrm{M}_c$ has a unique
minimiser, and no flat segment can occur, consistent with the strict decrease
Proposition~\ref{prop:reg} proves in the unconstrained case. Flat segments need a
constraint normal aligned with the support's signs, $Z\trans s_\F=0$, and so live in the
constrained, rank-deficient regime. They need the rank deficiency as well as the
alignment: under the standing assumption $Z\trans s_\F=0$ gives
$\delta_\F=-Z(Z\trans H_{\F\F}Z)^{-1}Z\trans s_\F=0$, so $\beta$ is constant on the
segment and $\beta(\lambda_1)=\beta(\lambda_2)$. A constant segment is a single point of
$\R^n$, not a flat segment in the sense of Proposition~\ref{prop:flat}, and it is how a
constrained path can start (Proposition~\ref{prop:ends}). Section~\ref{sec:conclusion}
returns to flat segments.

\begin{proposition}[Where the path starts and ends]
\label{prop:ends}
Let the feasible set $\mathcal P=\{\beta:A\beta=b,\ G\beta\le h\}$ of
\eqref{eq:lasso} be nonempty, write $R(\beta)=\tfrac12\|y-X\beta\|_2^2$, and let
$\beta(\lambda)$ be single-valued with finitely many breakpoints, as under the
hypotheses of Theorem~\ref{thm:identity}. Let $\mathcal P^\star$ be the set of points of
$\mathcal P$ of smallest $\ell_1$ norm, $m^\star$.
\begin{enumerate}
\item[(i)] There is a $\lambda_0<\infty$ such that $\beta(\lambda)$ is constant on
$[\lambda_0,\infty)$, equal to a minimiser of $R$ over $\mathcal P^\star$, and
$c(\lambda)=m^\star$ there.
\item[(ii)] As $\lambda\to0^+$, $\beta(\lambda)$ converges to a point of smallest
$\ell_1$ norm among the minimisers of $R$ over $\mathcal P$.
\end{enumerate}
Where the minimiser in (i), or the point in (ii), is unique, the path starts, or ends,
there.
\end{proposition}

\begin{proof}
For $\lambda>0$ and any $\hat\beta\in\mathcal P$, optimality of $\beta(\lambda)$ gives
\begin{equation}
\label{eq:ends}
R(\beta(\lambda))+\lambda\|\beta(\lambda)\|_1\le R(\hat\beta)+\lambda\|\hat\beta\|_1 .
\end{equation}
A convex quadratic bounded below on a nonempty polyhedron attains its infimum
\citep{frankwolfe1956}, so $R$ has a minimiser $\beta^\star$ on $\mathcal P^\star$, a
polyhedron, and a minimiser on $\mathcal P$. For (i) take $\hat\beta=\beta^\star$. As
$\|\beta(\lambda)\|_1\ge m^\star$, \eqref{eq:ends} gives $R(\beta(\lambda))\le
R(\beta^\star)$ and $\|\beta(\lambda)\|_1-m^\star\le R(\beta^\star)/\lambda$, so
$\beta(\lambda)$ is bounded as $\lambda\to\infty$. On the last segment
$[\lambda_0,\infty)$ it is $\alpha+\lambda\delta$ by Lemma~\ref{lem:affine}, and bounded
only if $\delta=0$. The constant $\bar\beta$ lies in $\mathcal P$, has $\ell_1$ norm
$m^\star$ by the bound, and has $R(\bar\beta)\le R(\beta^\star)$, so it minimises $R$
over $\mathcal P^\star$. For (ii) take $\hat\beta$ to be a minimiser of $R$ over
$\mathcal P$ of smallest $\ell_1$ norm, $m_0$, which exists because those minimisers
form a polyhedron. As $R(\beta(\lambda))\ge R(\hat\beta)$, \eqref{eq:ends} gives
$\|\beta(\lambda)\|_1\le m_0$ and $R(\beta(\lambda))\le R(\hat\beta)+\lambda m_0$. On
the first segment $(0,\lambda_1]$ the path is affine and so has a limit $\beta_0$ as
$\lambda\to0^+$; it lies in $\mathcal P$, minimises $R$ there, and has $\ell_1$ norm at
most $m_0$.
\end{proof}

Part (i) replaces $\lambda_{\max}$. Under homogeneous constraints $\beta=0$ is feasible,
$\mathcal P^\star=\{0\}$, and $\lambda_0$ is $\lambda_{\max}$. Under rows that exclude
$\beta=0$, $\mathcal P^\star$ can be a face rather than a point: with $\beta\ge0$ and
$\mathbf 1\trans\beta=1$, the fully invested constraint of sparse portfolio regression
\citep{brodie2009}, it is the whole simplex. The path then starts at the least-squares
fit over that face, which the linear program for $m^\star$ alone does not find.
\citet{hobson2026classo} initialise their path-follower in exactly this lexicographic
way, the $\ell_1$ norm first and the loss second, under $H\succ0$; part (i) gives the
same start under the standing assumption and adds that the path stays there until
$\lambda_0$. On $[\lambda_0,\infty)$ the slope is $\delta=0$, so $c'=0$ and, by
Proposition~\ref{prop:reg}, $Z\trans s_\F=0$: the path starts on a constant segment.
Part (ii) is the constrained form of the unconstrained statement that the LASSO path ends
at the least-squares fit of smallest $\ell_1$ norm \citep{tibshirani2013}. With more
predictors than observations that fit typically interpolates, and the reduced-Hessian
condition can hold all the way to it while $H$ is singular. A ridge $\epsilon I$, added to
make $H\succ0$, moves the end to the fit of the ridged problem instead.
\citet{schmelzer2026constrained} traces such paths with no ridge and measures how far
ridged path algorithms end from them.

\begin{remark}[Uniqueness, and what it costs]
\label{rem:unique}
General position is a real hypothesis and not a courtesy. Under it the unconstrained
LASSO solution is unique at every $\lambda$ \citep{tibshirani2013}; past it the same
paper shows the solution set can be a face rather than a point, and no counterpart is
known for the constrained problem \eqref{eq:mark} (Remark~\ref{rem:rank}, and
Section~\ref{sec:conclusion} for what that leaves open). What
Theorem~\ref{thm:identity} adds is the linear-independence constraint qualification at
each visited active set, which is what makes the constrained path single-valued and its
breakpoint sequence finite.
\end{remark}

\subsection{Budget and tilt}
\label{sec:bridge}

Theorem~\ref{thm:identity} and the map above relate $\mathrm{L}_\lambda$ to
$\mathrm{M}_c$, and $\mathrm{M}_c$ holds the linear term fixed while the budget $c$
moves. The frontier of
Section~\ref{sec:frontier-primer} is the other one-parameter family: the budget is fixed
by the mandate and the \emph{tilt} $\lambda$ moves. It is also the sweep a
practitioner runs; mandates fix a leverage budget, they do not sweep it. These are not
the same family pointwise, and the bridge between them is a rescaling rather than an
identity of parameters. What that rescaling has to carry depends on whether the
constraints other than the cap are homogeneous. Where they are, it is a scalar, and the
second family is the first one read backwards.

\begin{corollary}[Fixed leverage, swept tilt]
\label{cor:tilt}
Let $\Sig=X\trans X$ and $\mu=X\trans y$, and let $b=0$, $h=0$, so that the feasible set
of $\mathrm{M}_c$ is the $\ell_1$ ball met with a cone. Fix a budget $c>0$ and consider
the tilted program
\begin{equation}
\label{eq:tiltprog}
\begin{aligned}
\mathrm{T}_\lambda^{c}:\quad \min_{w}\ &\tfrac12\, w\trans\Sig\,w-\lambda\,\mu\trans w\\
\text{s.t.}\quad &\|w\|_1\le c,\ \ Aw=0,\ \ Gw\le 0 .
\end{aligned}
\end{equation}
Then for every $\lambda>0$, $w$ solves $\mathrm{T}^c_\lambda$ if and only if
$\lambda^{-1}w$ solves $\mathrm{M}_{c/\lambda}$, which assumes nothing beyond the
substitution. Add the
hypotheses of Theorem~\ref{thm:identity}, so that both sides are single-valued, and
write $w^{\mathrm M}(c)$ for the path of $\mathrm{M}_c$ and $w^{\mathrm T}_c(\lambda)$
for that of $\mathrm{T}^c_\lambda$; then
$w^{\mathrm T}_c(\lambda)=\lambda\,w^{\mathrm M}(c/\lambda)$, so the tilted family and
the capped family are one curve up to positive rescaling, agreeing as rays through the
origin rather than as point sets. The one has a corner exactly where the other does,
and the breakpoints correspond under $\lambda\mapsto c/\lambda$. Equivalently, the two parameters enter
only through their ratio, up to the scalar $\lambda$: the efficient set over
$(c,\lambda)$ is the cone over one path in $\R^n$, and in (return, standard deviation,
gross exposure) coordinates the surface of Section~\ref{sec:frontier-primer} is the
cone over one frontier.
\end{corollary}

\begin{proof}
Fix $\lambda>0$ and substitute $w=\lambda v$ in \eqref{eq:tiltprog}. The objective
becomes $\lambda^2\bigl(\tfrac12 v\trans\Sig v-\mu\trans v\bigr)$; the cap
$\|w\|_1\le c$ becomes $\|v\|_1\le c/\lambda$; and $Aw=0$, $Gw\le0$ become $Av=0$,
$Gv\le0$, the right-hand sides being fixed points of the scaling. The constraint data
is therefore that of $\mathrm{M}_{c/\lambda}$ exactly, and as $\lambda^2>0$ the
objectives differ by a positive factor, so the feasible sets and the orderings on them
correspond and the minimisers do too. The stated equivalence is that correspondence
read off; $w^{\mathrm T}_c(\lambda)=\lambda\,w^{\mathrm M}(c/\lambda)$ is the same
statement once
Theorem~\ref{thm:identity} makes each side a point, which the substitution does not
disturb, being a bijection of $\R^n$ carrying active sets to active sets.
\end{proof}

So the LASSO path already holds the fixed-leverage sweep: each vertex $\beta(\lambda)$
of $\mathrm{L}_\lambda$, rescaled to gross exposure $c$, is a corner of $\mathrm{T}^c$,
at tilt $c/\|\beta(\lambda)\|_1$. The sweep has no low-risk end. The cap goes slack
below some tilt whenever the cap-free program has a solution of gross exposure under
$c$, and below that tilt $\mathrm{T}^c_\lambda$ follows the ray through that solution
into the origin, where there is nothing for a homotopy to find. A
minimum-variance end needs a row that holds the portfolio invested, and such a row is
inhomogeneous, which is the case the corollary excludes and the rest of this section
treats. A budget $\mathbf 1\trans w=1$ becomes $\mathbf 1\trans v=1/\lambda$ under the
same substitution, a position limit $w\le\ub$ becomes $v\le\ub/\lambda$: the feasible
set travels with the parameter rather than standing still.
Proposition~\ref{prop:bridge} is where that is done.

\begin{proposition}[The frontier is a nonnegative LASSO path]
\label{prop:bridge}
Let $\Sig=X\trans X$ and $\mu=X\trans y$, let $0\le\lb\le\ub$ be a box for which
$\{w:\mathbf 1\trans w=1,\ \lb\le w\le\ub\}$ is nonempty, and let $w(\lambda)$ trace the
fully invested frontier, i.e.\ the minimiser of
$\tfrac12 w\trans\Sig w-\lambda\mu\trans w$ subject to $\mathbf 1\trans w=1$,
$\lb\le w\le\ub$. Assume that minimiser is unique at each $\lambda\ge0$, as it is when
the standing assumption holds at the active set it occupies and every active row and
blocked coordinate carries a nonzero multiplier, since a strict local minimiser of a
convex program is its only one. Write $t=1/\lambda$ and let $v(t)$ solve
\begin{equation}
\label{eq:nnls}
\min_v\ \tfrac12\|Xv-y\|_2^2 \quad\text{s.t.}\quad
\mathbf 1\trans v=t,\ \ t\,\lb\le v\le t\,\ub .
\end{equation}
Then $w(\lambda)=v(t)/t$ for every $\lambda>0$. Since $\lb\ge0$ forces $v\ge0$ and hence
$\mathbf 1\trans v=\|v\|_1$, \eqref{eq:nnls} is the nonnegative LASSO in the
bound-indexed form of $\mathrm{M}_c$ with $c=t$, its budget row holding $\|v\|_1$
\emph{at} $t$ where $\mathrm{M}_c$ caps it. That is the same constraint wherever the
cap binds, which is where the multiplier $\nu$ of the budget row is nonnegative. Where
$\nu<0$ the budget row still binds, as an equality may and a cap or a penalty may not,
and there \eqref{eq:nnls} is the nonnegative LASSO of the shifted response $y-\nu_0d$,
with $X\trans d=\mathbf 1$ and $\nu_0$ below every value $\nu$ takes at a corner. Both its budget row and its box travel
with the parameter, which is why \eqref{eq:pqp} lets the constraint data be affine in
$\lambda$ and not merely the objective. The frontier and that path are therefore the same curve
\emph{up to the positive rescaling} $v\mapsto v/\|v\|_1$: they agree as rays through the
origin, not as point sets. As $\lambda\downarrow0$ the correspondence closes by a limit,
$v(t)/t\to w_{\mathrm{mv}}$, the global minimum-variance portfolio.
\end{proposition}

\begin{proof}
With $\Sig=X\trans X$ and $\mu=X\trans y$,
\[
\begin{aligned}
\tfrac12 w\trans\Sig w-\lambda\mu\trans w
&= \tfrac12\|Xw\|_2^2-\lambda\,\langle Xw,y\rangle\\
&= \tfrac12\|Xw-\lambda y\|_2^2-\tfrac{\lambda^2}{2}\|y\|_2^2 .
\end{aligned}
\]
Substituting $w=\lambda v$ multiplies the first term by $\lambda^2$, giving
$\lambda^2\big(\tfrac12\|Xv-y\|_2^2\big)$ up to the same additive constant, while
$\mathbf 1\trans w=1$ becomes $\mathbf 1\trans v=1/\lambda=t$ and $\lb\le w\le\ub$
becomes $t\,\lb\le v\le t\,\ub$. As $\lambda^2>0$ the minimisers correspond, so
$w(\lambda)=\lambda\,v(t)=v(t)/t$. With the cap $\|v\|_1\le t$ in place of the budget
row, the optimality conditions are the same except that the multiplier must be
nonnegative. Under the constraint qualification the multipliers of \eqref{eq:nnls} are
unique, so $v(t)$ solves the capped problem exactly when $\nu\ge0$. For the shift, $X\trans
d=\mathbf 1$ gives
\[
\begin{aligned}
\tfrac12\|Xv-y\|_2^2+\nu\,\mathbf 1\trans v
&=\tfrac12\|Xv-(y-\nu_0d)\|_2^2\\
&\quad+(\nu-\nu_0)\,\mathbf 1\trans v+\text{const},
\end{aligned}
\]
so the penalty $\nu-\nu_0$ of the shifted problem is nonnegative wherever $\nu\ge\nu_0$.
For the limit, $w(\lambda)$ minimises
$\tfrac12 w\trans\Sig w-\lambda\mu\trans w$ over the fixed polytope
$\{\mathbf 1\trans w=1,\ \lb\le w\le\ub\}$, which is compact, and on a compact set
that objective converges uniformly to $\tfrac12 w\trans\Sig w$ as $\lambda\downarrow0$.
Uniform convergence of continuous objectives on a compact set carries minimisers to
minimisers of the limit, unique here by assumption, so
$w(\lambda)\to w_{\mathrm{mv}}$.
\end{proof}

Four readings follow. First, the frontier's parameter is a budget in disguise:
the gross-exposure ball is indeed vacuous at a \emph{fixed} budget of $1$, and
Proposition~\ref{prop:bridge} says what replaces it, so
the long-only frontier is not a member of the family for which the penalty has nothing
to do. Second, nothing on the left of \eqref{eq:nnls} is fixed, since the substitution
divides through by $\lambda$ and takes the constraints with it; had we fixed the box,
the proposition would hold only where $\ub$ is slack, and the Critical Line Algorithm's
reason for being is precisely the binding case. Third, the correspondence is radial
rather than pointwise. Breakpoint counts, the order of events and which coordinates are
free at each corner are invariant under $v\mapsto v/\|v\|_1$ and transfer verbatim;
anything metric in $\R^n$ does not, and the minimum-variance end is a limit rather than
a point of the budget-indexed path, since $\lambda=0$ is $t=\infty$. Fourth, a LASSO
solver takes no negative penalty. In the long-only case, $\lb=0$ and $\ub=\infty$,
$\nu$ reaches zero at $v_0=\arg\min_{v\ge0}\frac12\|Xv-y\|_2^2$. Writing $v=sw$ with
$\mathbf 1\trans w=1$ and minimising over $s\ge0$ first leaves
$-\frac12(\mu\trans w)_+^2/w\trans\Sig w$, so $v_0/\|v_0\|_1$ is the long-only
maximum-Sharpe portfolio. A solver run on $(X,y)$ stops there, and only the shift
reaches the frontier beyond it. That is the sense, and the only sense, in which the
turning points of Figure~\ref{fig:frontier} are LASSO breakpoints. The first $17$ of the
$20$ are breakpoints of the nonnegative LASSO path of $(X,y)$. The next two lie where
$\nu<0$ and are breakpoints only of the shifted path, and the last, the minimum-variance
portfolio, is its limit.

\paragraph{Bound versus penalty}
\label{par:hinge}
In $\mathrm{M}_c$ sparsity is governed by an \emph{explicit} constraint, the
gross-exposure ball; in $\mathrm{L}_\lambda$ it is \emph{induced} by a penalty through
its subgradient. They coincide except in one place, where the explicit one goes vacuous.
For a long-only, fully invested portfolio ($w\ge0$, $\mathbf 1\trans w=1$) we have
$\|w\|_1=\mathbf 1\trans w=1$ at every feasible $w$, so $\mathrm{M}_c$ is infeasible for
$c<1$, holds with equality but redundantly at $c=1$, and is slack for $c>1$: the cap
never shapes the solution, and the identity has nothing to bite on. Short positions
under a binding cap $c>1$ make it active, and $\mathrm{M}_c$ and $\mathrm{L}_\lambda$
coincide as in Theorem~\ref{thm:identity}. What governs the long-only case instead is
Proposition~\ref{prop:bridge}, where the budget moves with the tilt.

\begin{remark}[What is prior art, and what is ours]
\label{rem:precedents}
It helps to separate the claims that could be made here, because they are not equally
available.

\emph{That the family is one family} is prior art, and more squarely than is comfortable.
\citet{gartner2012} name the LASSO, least angle regression, the support vector machine,
model predictive control and mean--variance portfolio selection in one sentence as
instances of a single parametric quadratic program, and state in the next the affine-data
rule of Section~\ref{sec:scheme}, citing \citet{ritter1962}. We claim no part of it.
\citet{gartner2009}
give a generic path method for the same program, applied to conjoint analysis rather than
to a portfolio, and \citet{mairal2012} repeat the observation.

The limit is one of route rather than of care. All three arrive through the statistical
path literature, \citet{gartner2009} building on the support-vector-machine path of
\citet{hastie2004svm} and \citet{mairal2012} on \citet{osborne2000} and
\citet{efron2004}; coming that way, mean--variance selection is met as a problem class
the framework covers, and \citet{markowitz1952}, which poses it, is the natural
citation. The 1956 algorithm lies off that path, and nothing in the question they were
asking would have sent them to it. So when \citet{mairal2012} cite Markowitz beside
\citet{osborne2000} and \citet{efron2004} for the piecewise linearity of the path, the
instinct is exactly right and the reference points at the paper that poses the problem
rather than the one that solves it; and what \citet{gartner2012} state is the accurate
and weaker thing, that the mean--variance \emph{problem} is a parametric quadratic
program.

\emph{That the constrained path can be traced} is prior art too. \citet{gaines2018} give
the constrained-LASSO path under general linear equalities and inequalities, by an
active-set homotopy on the same KKT system. Their parameter is the penalty and their
constraint data is fixed, though: the path runs from $\rho_{\max}$ down to $\rho=0$ and
stops there. Under a fully invested mandate $\|\beta\|_1$ is constant on the feasible
set, so that parameter does no work at all and the curve is swept by the budget row
instead, which is Proposition~\ref{prop:bridge} and not something their path reaches. So is the portfolio reading of an $\ell_1$
constraint: \citet{brodie2009} recognised the Markowitz problem as a LASSO, and a
literature around it reads the constraint statistically, \citet{fan2012} treating the
gross-exposure-constrained frontier in the vocabulary this section borrows. What that
literature establishes is what the constraint \emph{does}; what it does not do is trace
the frontier as a path.

\emph{That the two paths are one curve} is what we add, and the division is clean. Prior
work establishes the common parametric-QP framework, of which we claim no part; what
follows here is an equivalence of the two solution \emph{paths}.
Theorem~\ref{thm:identity} shows that Markowitz's Critical Line Algorithm is an instance
of the same active-set homotopy and that, under the substitution, its path coincides with
the constrained LASSO path, under arbitrary linear equality and inequality constraints.
Propositions~\ref{prop:reg} and~\ref{prop:flat} pin down the parameter map
$c(\lambda)=\|\beta(\lambda)\|_1$, match the breakpoints one-to-one, and give the single
way that correspondence fails. Proposition~\ref{prop:bridge} supplies the rescaling the
frontier needs, its parameter being a tilt and not a budget. These are claims about two
procedures and their output; the prior art makes claims about two problem classes, which
is a coarser thing and does not imply them. \citet{brodie2009} swept a leverage penalty
at fixed return rather than tracing the risk--return frontier, so their curve and the
frontier meet only at the zero-penalty end, and \citet{gaines2018} never read the
constrained path as a portfolio. The unconstrained slice is the classical LARS/LASSO
correspondence \citep{efron2004}, and piecewise linearity of a quadratic loss under a
polyhedral penalty is classical \citep{rosset2007}.

\emph{What the identity does not carry} is ours as well, and we would not trade it for
the theorem. Section~\ref{sec:inference} says why the inferential half does not travel,
and then crosses that boundary deliberately. None of the precedents above had occasion to
raise the question, because none of them paired the two problems.

\end{remark}

\section{What the identity does not carry, and one deliberate crossing}
\label{sec:inference}
The \emph{objects} a statistician reads along the LASSO path all have portfolio names,
while the \emph{inference} does not come free. The reason is that the correspondence
between the two parametrisations, $c(\lambda)=\|\beta(\lambda)\|_1$, is computed from the
data. Statements about the curve transfer, because the curve does not know which name it
is called by; statements that average over the response do not, because holding $\lambda$
fixed and holding $c$ fixed are then different experiments. The identity is one of
curves, not of statistical experiments. Degrees of freedom are the clearest case, and the one where the
boundary can be crossed rather than merely reported. The active support is exactly the
set of assets held away from a bound, so as a \emph{descriptive} complexity measure the
count transfers at once, and the portfolio literature had no occasion to define it. The
\emph{distributional} statement does not follow from the penalised one. That $|\F|$ is
unbiased for the degrees of freedom is a theorem about the estimator at fixed $\lambda$
\citep{zou2007df, tibshirani2012df}, the second reaching it by the same Stein route we
use, whereas the frontier is indexed by $c$ and $c(\lambda)$ depends on $y$. This is an
extension of that accounting, not a separate discovery, and most of the extending has
been done. \citet{kato2009} settles the $\ell_1$-\emph{constrained} problem, which is the
indexing the frontier needs, but without further constraints.
\citet{hu2015} settle the linearly constrained \emph{penalised} problem, and
\citet[eq.~18]{gaines2018}, following them, state the count in the form nearest to ours:
at fixed penalty $\rho\ge0$, subject to $A\beta=b$ and $C\beta\le d$,
$\mathrm{df}=\mathbb E\big[|\mathcal A|-(q+|\mathcal Z_I|)\big]$, the active predictors
less the equalities and the binding inequalities; see also \citet{zhoulange2013}, whose
constrained path the count agrees with. The accounting of
Proposition~\ref{prop:df} is theirs; what is added is the two respects in which the
frontier's problem is not the one they solve.

The first is the indexing. Their tuning parameter is the penalty, with the constraint
data $b$ and $d$ held fixed, so no row of theirs plays the part of the ball; ours is the
bound $c$, and the ball's sign row $s_\F\trans$ enters $M_\F$ as a constraint like any
other. The distinction is not cosmetic here: under a fully invested mandate
$\|\beta\|_1$ is constant on the feasible set, the penalty does no work at all, and the
curve is swept by $b$ rather than by $\rho$ (Proposition~\ref{prop:bridge}). The second
is that a count becomes a rank. Where they subtract $q+|\mathcal Z_I|$ we subtract
$\operatorname{rank}M_\F$, which agrees with the count exactly when the active rows are
independent and is smaller when they are not. A mandate makes them dependent routinely:
a budget row, a sector equality summing to the same weights, and a position limit
binding at a bound are not in general position. Proposition~\ref{prop:df} supplies both,
and the bordered system of Section~\ref{sec:scheme} is why it costs nothing more: the
same Schur complement that gives the segment gives the divergence.

The geometry is simpler than the algebra, and worth having first. Whatever the
constraints, $\hat y$ is the Euclidean projection of $y$ onto
$K=\{X\beta:\beta\ \text{feasible}\}$, a closed convex polyhedron: minimising
$\|y-X\beta\|_2$ over a polyhedron in $\beta$ is exactly projecting $y$ onto its image.
A projection onto a polyhedron is, locally, the projection onto whichever face carries
it, so its divergence is the dimension of that face. Counting that dimension is the whole
of what follows: the support contributes $|\F|$ directions, and each independent active
row (the ball's, an equality's, a binding inequality's) removes one. Where the
three results above differ from ours is only in how many rows there are to remove.

\begin{proposition}[Degrees of freedom under general constraints]
\label{prop:df}
Let $y\sim N(\mu,\sigma^2 I_m)$ and let $\hat\beta$ solve $\mathrm{M}_c$ of
\eqref{eq:mark} with $X$, $c$ and the constraint data $A,b,G,h$ fixed; write
$\hat y=X\hat\beta$. Suppose that for almost every $y$ the minimiser is unique, the
active set is locally constant, and $X_\F$ has full column rank on the support
$\F$. Collect the active constraint rows on the support,
\[
M_\F=\begin{bmatrix} s_\F\trans\\ A_\F\\ G_{\Sset\F}\end{bmatrix},
\]
the first row present only when $\|\hat\beta\|_1=c$. Then $\hat y$ is almost
differentiable, $\partial\hat y/\partial y$ is the orthogonal projection onto the column
space of $X_\F Z_M$, the columns of $Z_M$ a basis for $\ker M_\F$ (the construction of
Proposition~\ref{prop:reg} with the ball's sign row adjoined to the active normals), and
\begin{equation}
\label{eq:df}
\mathrm{df}(\hat y)\;:=\;\frac{1}{\sigma^2}\sum_{i=1}^m\operatorname{Cov}(\hat y_i,y_i)
\;=\;\mathbb E\big[\,|\F|-\operatorname{rank}M_\F\,\big].
\end{equation}
\end{proposition}

\begin{proof}
On the support the KKT conditions of $\mathrm{M}_c$ read
$X_\F\trans(X_\F\hat\beta_\F-y)+M_\F\trans\pi=0$ together with
$M_\F\hat\beta_\F=(c;\,b;\,h_\Sset)$, where $\pi$ stacks the multipliers of the
ball, the equalities and the active inequalities. This is \eqref{eq:border} with
$M_\F$ in the border, and only its top right-hand side depends on $y$, through
$X_\F\trans y$. Differentiating and eliminating $\pi$ by the Schur complement, exactly
as in the proof of Proposition~\ref{prop:reg}, gives
\[
\partial\hat\beta_\F/\partial y
  = Z_M\big(Z_M\trans X_\F\trans X_\F Z_M\big)^{-1}Z_M\trans X_\F\trans .
\]
Writing $W=X_\F Z_M$, the fit therefore obeys
\[
\partial\hat y/\partial y = W(W\trans W)^{-1}W\trans ,
\]
the projection onto $\operatorname{col}W$, whose trace is
$\operatorname{rank}W$. Since $X_\F$ has full
column rank, $\operatorname{rank}W=\dim\ker M_\F=|\F|-\operatorname{rank}M_\F$.
The map $y\mapsto\hat y$ is the Euclidean projection of $y$ onto the closed convex
polyhedron $\{X\beta:\beta\ \text{feasible}\}$ and so is $1$-Lipschitz, hence almost
differentiable with bounded derivative; Stein's identity therefore applies and
$\mathrm{df}(\hat y)=\mathbb E[\operatorname{div}\hat y]$, which is the stated
expectation.
\end{proof}

\begin{figure}[t]
\centering
\includegraphics[width=0.98\columnwidth]{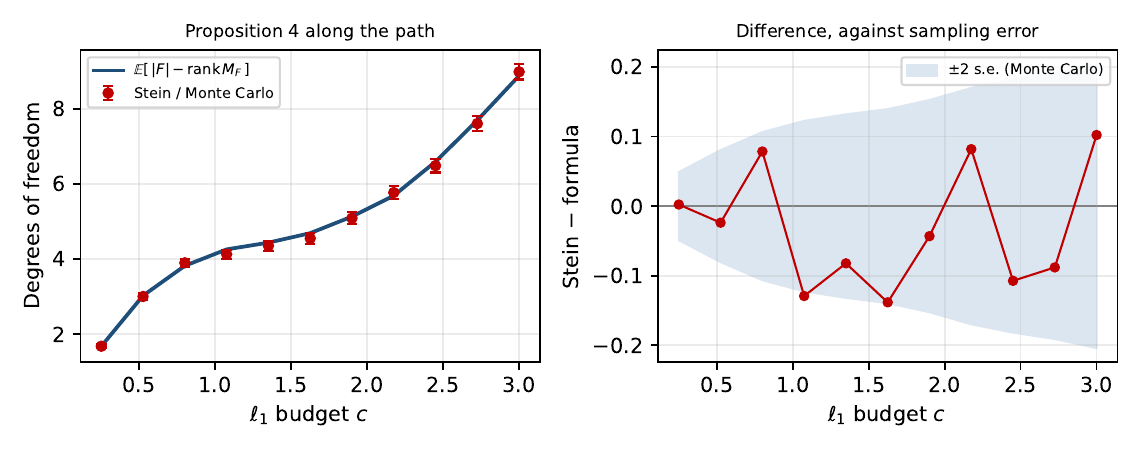}
\caption{Proposition~\ref{prop:df} checked along the path. \emph{Left:} degrees of
freedom at eleven budgets, by $\mathbb E[|\F|-\operatorname{rank}M_\F]$ (line) and by
Monte-Carlo estimation of Stein's $\sigma^{-2}\sum_i\operatorname{Cov}(\hat y_i,y_i)$
over $1500$ draws (points, two-standard-error bars); \emph{right:} their difference,
against the same band. The ball and an inequality are both active over most of the
range, so the count is neither $|\F|$ nor $|\F|-1$.}
\label{fig:df}
\end{figure}

\begin{remark}[Which hypotheses matter]
Three of the four are conveniences of the accounting rather than of the result.
Uniqueness of $\hat\beta$ is not needed for $\hat y$, which is a projection onto a
convex set and so is unique always; it is needed only so that $\F$ and $M_\F$ are well
defined. Local constancy of the active set fails on a set of $y$ of measure zero, which
is where almost differentiability already tolerates failure. Full column rank of $X_\F$
is what turns $\operatorname{rank}(X_\F Z)$ into $|\F|-\operatorname{rank}M_\F$;
without it the divergence is still $\operatorname{rank}(X_\F Z)$, but the count loses
its reading as holdings less binding constraints. What is genuinely necessary is the
Lipschitz property, and that is free: a projection onto a convex set is $1$-Lipschitz
whatever the data.
\end{remark}

Three readings, in increasing order of interest. With $A$ and $G$ absent and the ball
active, $\operatorname{rank}M_\F=1$ and \eqref{eq:df} recovers Kato's
$\mathrm{df}=\mathbb E|\F|-1$: the constrained form spends one degree of freedom on the
budget that the penalised form does not. With the ball slack and no other constraints it
is $\mathbb E|\F|$, least squares on the selected support. And for the long-only, fully invested frontier the support carries $s_\F=
\mathbf 1$, so the ball row and the budget row $\mathbf 1\trans$ \emph{coincide},
$\operatorname{rank}M_\F=1$ however tight the mandate, and
\[
\mathrm{df}=\mathbb E\big[\,\#\{\text{holdings away from a bound}\}\,\big]-1 .
\]
A frontier portfolio's effective dimension is its holding count, less one for the budget.
That is the statement the count of Section~\ref{sec:utility} was waiting for, and it is
the one transfer in this note that runs from statistics into portfolio selection and
arrives with a theorem rather than an analogy. Equation~\eqref{eq:df} holds not at a point but along the
path, and Figure~\ref{fig:df} checks it there. Both sides are estimated at eleven
budgets on an $n>m$ design: $m=20$, $n=30$, a three-term truth, $\sigma=0.5$, and one
inequality row $\mathbf 1\trans\beta\le0.6$ chosen so that the ball and the inequality
both bind, which is the general-constraint case \citet{kato2009} does not cover. The
degrees of freedom climb from $1.7$ to $8.9$ as the budget opens, and the count tracks
the Monte-Carlo estimate the whole way: the mean discrepancy is $0.03$ against a standard
error of $0.02$, and ten of the eleven budgets agree within two standard errors. The fit
is solved afresh by a conic solver at every draw, so the two sides share no machinery. Post-selection inference divides the
same way. The polyhedron that \citet{lee2016} and \citet{tibshirani2016psi} condition on
is, under the identity, the set of data for which a given corner portfolio is optimal, so
the geometry transfers; the conditioning event does not, since a fixed $\lambda$ and a
fixed gross-exposure cap are different experiments.

So the identity is exact, and what it carries is the curve: its corners, the order they
come in, the algorithm that finds them. That is less than one might hope for and more
than it sounds, because a curve is what a path-following method computes, and the methods
that compute one are a larger family than the two programs paired here, which
\citet{gartner2012} and \citet{gartner2009} lay out and we take as given.

\section{History}
\label{sec:history}

The device was invented three times over fifty years, in four literatures, and the
arrivals had almost nothing to do with one another.\footnote{This section is a
sketch, not a survey; each of the four literatures could fill an article of its own.}
Optimization and finance
built it together in the 1950s; statistics reached it by a different road four decades
later, asking a different question; control found it again in the 2000s, and only there
did anyone arrive already knowing where the others had been.

\subsection{The Critical Line Algorithm}
\label{sec:cla-history}

The active-set treatment of the parametric quadratic program is a creation of the
1950s, and the Critical Line Algorithm is its earliest fully worked
instance. Every ingredient was to
hand in that decade: Dantzig's simplex method \citep{dantzig1951}, whose vertex-to-vertex
walk the homotopy still rides; the optimality conditions of \citet{kuhntucker1951},
anticipated by \citet{karush1939}, whose parametric solution is the bordered KKT system of
Section~\ref{sec:scheme}; the \emph{parametric} linear programming of
\citet{gasssaaty1955}, which follows optimal vertices as a cost coefficient varies and
becomes the critical line once the objective is quadratic; and the general
quadratic-programming algorithms of \citet{frankwolfe1956}, \citet{beale1959}, and
\citet{wolfe1959}, several appearing alongside Markowitz's in the \emph{Naval Research
Logistics Quarterly}. Markowitz worked alongside Dantzig at the RAND Corporation and took
the simplex method as his starting point: having posed mean--variance selection
\citep{markowitz1952}, he recognised that the efficient set is obtained by solving a
quadratic program \emph{parametrically} in the risk--return trade-off, and
\citet{markowitz1956} gave the Critical Line Algorithm, circulated as a RAND
memorandum in February 1955 \citep{markowitz1955} and developed in full in the 1959
monograph \citep{markowitz1959}. The method is an instance of the homotopy of
Section~\ref{sec:scheme}: trace the critical line, changing the status of one asset at
each turning point, contemporary with and independent of the general theory that would
later be seen to contain it. That theory, which is Lemma~\ref{lem:affine}, arrived six years afterwards in a German
operations-research journal \citep{ritter1962}, and reached the journal that had carried
Markowitz only in 1967, in M.~Meyer's English translation, eleven years
after Markowitz's paper. Markowitz had the special case first, and Ritter's general
case, when it came, did not mention him. Inside optimization the debt was acknowledged.
\citet{wolfe1959} took the critical line method as his starting point; Markowitz, in
Appendix~B of the 1959 monograph, held the two procedures equivalent (started
together, they continue together), and \citet{cottle2010} supply the proof, as
principal pivoting on a parametric linear complementarity problem. The algorithm has been
paired with a general method before, then, and rigorously; what it was never paired with
is the homotopy statistics would build. The parametric line runs on from there to
\citet{best1996}, who sets it out in the form closest to Section~\ref{sec:scheme}, in a
volume honouring Ritter.

\paragraph{What is already in RM-1438} The memorandum repays reading rather than
citation, and citation is mostly what it got: \citet{cottle2010} find it in two
bibliographies of the period and in the text of neither. It poses the problem with equalities, general inequalities and nonnegativity
rather than a budget alone; it assumes the quadratic form only positive
\emph{semi}-definite, with strict convexity on the equality-constrained set, which is
the reduced-Hessian condition of Section~\ref{sec:scheme}; it proves the bordered matrix
nonsingular by showing its kernel meets the constraint null space only at zero; it writes
the segment in closed form as $X_j=a_j+b_j\lambda_E$, which is \eqref{eq:affine}; and it
finds the next breakpoint as the largest $\lambda_E$ below the current one at which any
of four quantities vanishes: a free variable, a slack, an inequality multiplier, or a
blocked variable's reduced cost, which are
$\mathrm{P}_1,\mathrm{P}_2,\mathrm{D}_2,\mathrm{D}_1$ of Section~\ref{sec:events}. It
updates the inverse by one row and column instead of refactorising, citing
R.~A.~Fisher's sweep and the simplex method in one breath. It perturbs $\Sig$ to
$\Sig+eI$ where strict convexity fails, and bounds what that costs. And it states the
degeneracy problem of Section~\ref{sec:conclusion}, and says plainly that it has no
method for it. Little of Section~\ref{sec:scheme} is absent from this document.

It also records a contemporary arrival the later literature forgot. Markowitz points to a
problem ``of very similar structure analyzed independently'' by \citet{houthakker1953},
on expenditure under a quadratic utility, where the parameter enters a budget row rather
than the objective, and to \citet{dorfman1951} on a monopolist's quadratic profit.
Independent arrival at this structure begins, then, at the beginning.

Two threads then ran from this beginning for half a century, both inside finance. On the
modelling side, \citet{sharpe1963}, again at Markowitz's suggestion, introduced the
single-index model, writing the covariance as a diagonal matrix plus a rank-one term;
modern multi-factor risk models generalise this diagonal-plus-low-rank structure, and it
is precisely what treating $H$ as an operator (Section~\ref{sec:scheme}) exploits. On the
algorithmic side, the CLA was re-derived, refined, and scaled across decades, almost
always from within portfolio theory
\citep{perold1984,niedermayer2010,markowitz2021}.
The lineage is unbroken; what it never asked was whether the algorithm belonged to a
family larger than portfolio selection.

The method never had a wider breakthrough, and the reasons differ by era.
\citet{cottle2010} put the early one down to presentation: the phrase ``quadratic
programming'' appears in Markowitz's paper once, in a footnote, and eleven of its twelve
sections are about efficient points rather than about minimising a quadratic. The later
reason is that the method was easy to do without: a generic
solver dispatches a single efficient point in milliseconds, and one approximates the
frontier by re-solving as the return tilt varies. If finance itself overlooked the computational contribution, we can hardly fault statistics
for arriving at the same homotopy without reference to it.

\subsection{How statistics came to it independently}
\label{sec:stats-history}

Statistics arrived at the same device around forty years later, by a different door,
and at first without reference to any of this. The door was variable selection and
shrinkage, not optimization: the question was not how to trace a frontier but how the
LASSO of \citet{tibshirani1996} chooses, as its penalty relaxes, which coefficients to
admit, and the discovery was that it does so along a piecewise-linear path.
\citet{osborne2000} gave the homotopy explicitly; \citet{efron2004} gave least angle
regression, read the path's corners as the events of a single forward procedure, and
made the piecewise-linear path famous well beyond the LASSO itself.
The construction was then carried to neighbouring losses and penalties, and to
generalized linear models, where the path is no longer piecewise linear and has to be
followed approximately, which is the natural boundary of the exact scheme. \citet{hastie2004svm}
carried the same construction to the support vector machine, tracing its solution as the
cost parameter varies, and \citet{tibshirani2011} extended it to the generalized lasso,
with the fused lasso and trend filtering among its special cases, by a dual path of the
same kind. What had been an optimization fact about parametric quadratic programs was
being rebuilt, theorem by theorem, as a native part of regularised estimation.

The development reached even the governing principle, \citet{rosset2007} restating in the
field's own vocabulary the criterion optimization had carried since the 1950s. None of it
leaned on that lineage, and there was little reason it should have: the statistical
versions answered questions about selection, degrees of freedom and prediction that the
portfolio literature had never posed.

The bridge back was late and partial. \citet{brodie2009} and \citet{gaines2018} each
recovered a piece of it, as the introduction sets out, and each recovered something the
memorandum already had (the four event families and the ridge alike), which is no
reflection on work done independently seventy years later. Even so, the
observation travelled as a curiosity rather than as the statement that two large
algorithmic literatures describe one structure.

The same device was then found once more, in control, and this time the crossing was
partly deliberate. Through the 2000s the explicit solution of the constrained
linear-quadratic regulator was worked out by \citet{bemporad2002} and \citet{tondel2003},
in the multi-parameter form Remark~\ref{rem:mpc} describes; \citet{roll2008} then made
the link to the scalar case plain, explicitly generalising the LARS/LASSO construction.
Control is the lone instance in this story of a field that reached the device knowing the
others had been there.

\subsection{One structure, four disconnected literatures}
\label{sec:disconnect}

What is striking is not that the structure existed, since optimization has had the
general language for it all along, nor even that it went unnoticed, since it did not:
\citet{gartner2012}, \citet{gartner2009} and \citet{mairal2012} all record the lineage.
Their reach stops at the problem, though, and Remark~\ref{rem:precedents} says exactly
where. What is striking is how little followed. Four versions are still developed and
taught in parallel, each with its own vocabulary and its own canonical references, and a
practitioner of one can work a career without meeting the others. That is the gap this
paper is addressed to, and it is a gap in the cashing out, not in the noticing.

\section{What the whole path shows}
\label{sec:utility}

That the path is exact is not a numerical nicety. Gridding the penalty and interpolating
between grid points cuts the very corners, the breakpoints, that carry the
model-selection information; coordinate-descent and proximal solvers on a
$\lambda$-grid approximate the path in exactly this way, whereas the homotopy returns
the corners themselves. Every agreement this note claims (with conic and active-set
solvers, with \texttt{lars\_path}, and with the Stein estimate of \eqref{eq:df}) is
checked against an independent reference and reproduced by the source repository. What
the corners are then good
for is the rest of this section, and two of those readouts are objects every statistician
already reads, wearing portfolio names.

\emph{The weight profile.} Plot each weight as the frontier is traced and, under
$\Sig=X\trans X$, what appears \emph{is} the LARS/LASSO coefficient profile: the
picture statistics reads to see which variables a model admits and in what order, the
same lines relabelled. Read as a portfolio it is a membership and turnover map of the
efficient set, each position entering, growing and leaving, the turning points being the
entry and exit events.

\emph{The number of holdings.} Counting the positions held away from a bound at each
point of the frontier gives the portfolio's size as a function of risk, which by the
identity is the LASSO's support size. It falls in the familiar knee shape: many
holdings shed for little added risk, then a long flat tail of one or two concentrated
names. Because the count is the support size, that
corner is recognisably a \emph{model-selection} elbow, and
Proposition~\ref{prop:df} is what lets one do more than recognise it. Since
$|\F|-\operatorname{rank}M_\F$ is unbiased for the degrees of freedom of the
\emph{bound-indexed} fit, it can be read off at each turning point and substituted into
an information criterion: $\|y-\hat y\|_2^2+2\sigma^2(|\F|-\operatorname{rank}M_\F)$
scores every corner of the frontier for the price of the trace itself, with no
resampling, and on the long-only frontier the correction is simply the holdings less
one. Two honest limits. The criterion still needs $\sigma^2$ from outside; and because
the corners were located using $y$, it estimates prediction error rather than furnishing
a post-selection guarantee, which is the distinction Section~\ref{sec:inference} draws.
Independently of any criterion, an investor can read off, at
any target volatility, how many names the optimum will hold, or invert the reading and
ask what risk a given cardinality buys.

Neither readout needs the whole frontier: the band one cares about can be traced alone,
paying only for the turning points it contains.

\section{Conclusion}
\label{sec:conclusion}

That the parametric active-set homotopy is a computational primitive in its own right,
rather than an appendage of whichever model it serves, has been said before and we have
taken it as given. What we have tried to supply is the part that does not follow from
saying it: an exact statement of when two of these procedures trace the same curve, and
an account of what that does and does not license.

The dividend is transfer, and it is narrower than it looks. What moves freely are
statements about the curve (a uniqueness theorem, a path-length bound, the order in
which coordinates enter) and the software with them, since an implementation is a
statement about the curve too. What does not move is anything attached to a fixed value
of the tuning parameter: a post-selection interval, and in general any quantity averaged
over the response (Section~\ref{sec:inference}), which is much of what a statistician
wants. Degrees of freedom were the exception worth pursuing, and
Proposition~\ref{prop:df} recomputes them for the bound-indexed fit under general
constraints rather than transporting them from the penalised case, which is what a
deliberate crossing of that boundary costs. Reporting where the boundary runs, and
crossing it once, is as much of the contribution as the identity itself.

Two problems are left open, and both belong to the shared object rather than to either
instance. The first is the length of the path: near-linear in practice, exponential in
the worst case \citep{mairal2012}, with no tie-break rule known that keeps the corners
exact and carries a guarantee. The second is rank deficiency. Less of it is deficient for
the homotopy than a singular $H$ suggests: the standing assumption asks for a positive
definite Hessian only on the null space of the active rows, so with more predictors than
observations and constraint rows the path can run exactly, with no perturbation, to the
least-squares fit of smallest $\ell_1$ norm (Proposition~\ref{prop:ends}), and a ridge
added to avoid the question moves that end \citep{schmelzer2026constrained}. What
remains open is where even the reduced Hessian is singular, where the standing
assumption bites (Remark~\ref{rem:rank}) and the literature divides between perturbing
the problem and selecting a particular minimiser. The perturbation is older than its
modern users suppose, \citet{markowitz1955} adding $eI$ and bounding what it costs. The
selection, at least, is settled by the structure rather than by any algorithm: the
constrained minimisers form a polytope on which both the fit and the $\ell_1$ norm are
constant, so the minimum-$\ell_2$-norm choice exists and is unique at every $\lambda$.
Whether any homotopy traces it is the most interesting question we leave open.

The primitive was Markowitz's first, in a 1956 paper cited far less than the 1952 one
it followed, and
statistics arrived at it again, by an entirely different road, half a century later. That
an idea is reached more than once, independently, is less an oversight to correct than a
measure of how good it is. It has been named already, more than once. Showing what two of
its instances share, exactly, and handing over an engine that traces both, is the tribute
this note means to pay.

\appendix
\section{The ratio tests, the LARS step, and the loop}
\label{sec:ratios}

Section~\ref{sec:events} names the four events and says what each does to the active
set. The formulas are collected here, together with the specialisation that identifies
$\mathrm{D}_1$ with the entry step of least angle regression and the loop the two
together define.

\paragraph{The bordered solve} The coefficient matrix of \eqref{eq:border} is
nonsingular under the standing assumption. If $H_{\F\F}u+C_\F\trans v=0$ and $C_\F u=0$,
then $u\trans H_{\F\F}u=-(C_\F u)\trans v=0$, so $H_{\F\F}u=0$ as $H\succeq0$; with
$u=Zz$ this gives $Z\trans H_{\F\F}Z\,z=0$, hence $z=0$ and $u=0$, and then
$C_\F\trans v=0$ gives $v=0$ because $C_\F$ has full row rank. Conversely a nonzero $u$
with $H_{\F\F}u=0$ and $C_\F u=0$, or a nonzero $v$ with $C_\F\trans v=0$, lies in the
kernel, so for $H\succeq0$ the condition is necessary as well; $H_{\F\F}$ itself may be
singular. This fixes the sign
convention used throughout: the offset $\alpha$ solves \eqref{eq:border} with $\lambda$
set to $0$ in its right-hand side, the slope $\delta$ solves it with the coefficient of
$\lambda$ as right-hand side, and a segment is read off as $\alpha+\lambda\delta$ with
each term's own sign already absorbed. Note what the bordered block absorbs: $G_{\Sset}$
enters only as additional rows of $C$, indistinguishable from the equality rows $A$, and
the box bounds are carried by the partition $(\F,\B)$ and the fixed $x_\B$ rather than
by $C$ at all.

Writing each quantity as its value at $\lambda_k$ plus $\sigma t$ times its rate gives
the four ratios. Write $\min^{+}$ for the least strictly positive member of a set, $+\infty$ if there is
none. Two of the four confine an affine quantity to an interval whose \emph{endpoints
are themselves affine}, and they take the same two-sided form; the other two are
one-sided. The four are
\begin{equation}
\label{eq:ratios}
\begin{aligned}
t_{\mathrm{P}_1} &= \min_{i\in\F}\ {\min}^{+}
   \Big\{\tfrac{\ub_i-x_i}{\sigma(\dot x_i-\dot\ub_i)},\
         \tfrac{x_i-\lb_i}{\sigma(\dot\lb_i-\dot x_i)}\Big\},\\[2pt]
t_{\mathrm{D}_1} &= \min_{i\in\B}\ {\min}^{+}
   \Big\{\tfrac{\ub^{\rho}_i-\rho_i}{\sigma(\dot\rho_i-\dot\ub^{\rho}_i)},\
         \tfrac{\rho_i-\lb^{\rho}_i}{\sigma(\dot\lb^{\rho}_i-\dot\rho_i)}\Big\},\\[2pt]
t_{\mathrm{P}_2} &= \min_{j\notin\Sset}\ {\min}^{+}
   \Big\{\tfrac{h_j-g_j\trans x}{\sigma(g_j\trans\dot x-\dot h_j)}\Big\},\\[2pt]
t_{\mathrm{D}_2} &= \min_{j\in\Sset}\ {\min}^{+}
   \Big\{\tfrac{-\zeta_j}{\sigma\dot\zeta_j}\Big\},
\end{aligned}
\end{equation}
and the step is
$t^\star=\min\{t_{\mathrm{P}_1},t_{\mathrm{P}_2},t_{\mathrm{D}_1},t_{\mathrm{D}_2}\}$.
Here $\dot x=\delta$ and $\dot\zeta$ come from the same bordered solve, $\dot\rho$ from
stationarity, and $\dot\lb,\dot\ub,\dot h$ are the data's own rates, zero when that
datum does not move. The
primal pair $\mathrm{P}_1,\mathrm{P}_2$ are the ratio tests of the simplex method and
the dual pair $\mathrm{D}_1,\mathrm{D}_2$ their counterparts; what the display makes
plain is that $\mathrm{P}_1$ and $\mathrm{D}_1$ are the \emph{same} test, one on a
primal quantity against the box and one on a dual quantity against the cone, and that
each reduces to a fixed threshold exactly when the corresponding rate vanishes.

\begin{remark}[$\mathrm{D}_1$ is the LARS entry step]
\label{rem:larsstep}
Specialise \eqref{eq:ratios} to the LASSO of Remark~\ref{rem:lassocheck}, travelling
towards smaller $\lambda$ ($\sigma=-1$). There $\rho_i=x_i\trans(y-X\beta)$ is the
correlation of an inactive predictor with the residual, the confining interval is
$[-\lambda,\lambda]$, and both sides move. Writing $a_i=x_i\trans X_\F\delta_\F$ for the
rate at which that correlation changes, the boundary is reached when
$|\rho_i-t\,a_i|=\lambda_k-t$, so
\[
t_{\mathrm{D}_1}=\min_{i\in\B}\ \min_{\pm}\
\Big\{\tfrac{\lambda_k\mp\rho_i}{1\mp a_i}\ :\ \text{positive}\Big\},
\]
which is the entry step of least angle regression in the form \citet{efron2004} write it,
$\min^{+}\{(\hat C-\hat c_j)/(A-a_j),\,(\hat C+\hat c_j)/(A+a_j)\}$, up to their
normalisation of the equiangular direction. The moving threshold is the whole of the
difference from the fixed-bound case: set the denominators' $\mp a_i$ against a stationary
$1$ and one is back at $\mathrm{P}_1$.
\end{remark}

\begin{center}
\fbox{\begin{minipage}{0.92\columnwidth}
\small
\textsc{Parametric active-set homotopy}\\[2pt]
\emph{Input:} $H$ as an operator, $q$, $d$, $\tau_0$, $\tau_1$, the constraint data, a
starting
$\lambda_0$ and an optimal active set $(\F,\B,\Sset)$ there.\\[3pt]
\textbf{repeat}
\begin{enumerate}\itemsep0pt\parskip0pt
\item Solve \eqref{eq:border} twice against one factorisation to get the segment
      $x_\F(\lambda)$ and the multipliers $\zeta(\lambda)$, $\rho(\lambda)$.
\item Form the step to each event in $\mathrm{P}_1,\mathrm{P}_2,\mathrm{D}_1,
      \mathrm{D}_2$; let $t^\star$ be the smallest, Bland's rule breaking ties.
\item Move to $\lambda_k+\sigma t^\star$, record the turning point, and apply the
      attaining event's update to $(\F,\B,\Sset)$.
\item Update the factorisation by the one index that changed.
\end{enumerate}
\textbf{until} no event is reachable.\\[3pt]
\emph{Output:} the breakpoints, between which the path is affine by
Lemma~\ref{lem:affine}.
\end{minipage}}
\end{center}

% Author contributions. IMS asks for this on a journal submission; a preprint
% is under no such obligation, so both preprint routes drop it. \PREPRINT is
% the same switch that drops the journal banner -- see the comment at the top.
% To carry it on arXiv as well, delete this \ifdefined line and its \fi.
\ifdefined\PREPRINT\else
\begin{acks}[Author contributions]
T.~Schmelzer conceived the unifying perspective and wrote the manuscript.
T.~Hastie contributed Section~\ref{sec:identity}.
\end{acks}
\fi

\begin{acks}[Acknowledgments]
\emph{In memory of Harry Markowitz (1927--2023).}\par
\smallskip
We thank Ryan~J.\ Tibshirani for helpful comments on an earlier draft. The
remaining errors are ours.\par
\smallskip
The roots of this work were planted during the first author's sabbatical at Stanford
University. He is most grateful to his host, Stephen Boyd, and to the convex
optimization group (\texttt{cvxgrp}) for the setting in which it took shape; and to
Philipp Schiele, who shares his affection for the Critical Line Algorithm.
\par
\smallskip
This note was drafted and revised with the assistance of Anthropic's Claude;
the authors are responsible for all content.
\end{acks}

% References are managed with BibTeX (refs/references.bib) and the IMS
% author-year style (imsart-nameyear.bst).
\bibliographystyle{imsart-nameyear}
\bibliography{refs/references}

@article{beale1959,
  author  = {Beale, E.~M.~L.},
  title   = {On quadratic programming},
  journal = {Naval Research Logistics Quarterly},
  year    = {1959},
  volume  = {6},
  number  = {3},
  pages   = {227--243},
}

@article{bemporad2002,
  author  = {Bemporad, A. and Morari, M. and Dua, V. and Pistikopoulos, E.~N.},
  title   = {The explicit linear quadratic regulator for constrained systems},
  journal = {Automatica},
  year    = {2002},
  volume  = {38},
  number  = {1},
  pages   = {3--20},
}

@article{brodie2009,
  author  = {Brodie, J. and Daubechies, I. and {De Mol}, C. and Giannone, D. and Loris, I.},
  title   = {Sparse and stable {Markowitz} portfolios},
  journal = {Proceedings of the National Academy of Sciences},
  year    = {2009},
  volume  = {106},
  number  = {30},
  pages   = {12267--12272},
}

@incollection{dantzig1951,
  author    = {Dantzig, G.~B.},
  title     = {Maximization of a linear function of variables subject to linear inequalities},
  editor    = {Koopmans, T.~C.},
  booktitle = {Activity Analysis of Production and Allocation},
  series    = {Cowles Commission Monograph No.~13},
  pages     = {339--347},
  publisher = {Wiley},
  address   = {New York},
  year      = {1951},
}

@article{efron2004,
  author  = {Efron, B. and Hastie, T. and Johnstone, I. and Tibshirani, R.},
  title   = {Least angle regression},
  journal = {The Annals of Statistics},
  year    = {2004},
  volume  = {32},
  number  = {2},
  pages   = {407--499},
}

@article{frankwolfe1956,
  author  = {Frank, M. and Wolfe, P.},
  title   = {An algorithm for quadratic programming},
  journal = {Naval Research Logistics Quarterly},
  year    = {1956},
  volume  = {3},
  number  = {1--2},
  pages   = {95--110},
}

@article{hu2015,
  author  = {Hu, Q. and Zeng, P. and Lin, L.},
  title   = {The dual and degrees of freedom of linearly constrained generalized lasso},
  journal = {Computational Statistics \& Data Analysis},
  year    = {2015},
  volume  = {86},
  pages   = {13--26},
}

@article{zhoulange2013,
  author  = {Zhou, H. and Lange, K.},
  title   = {A path algorithm for constrained estimation},
  journal = {Journal of Computational and Graphical Statistics},
  year    = {2013},
  volume  = {22},
  number  = {2},
  pages   = {261--283},
}

@article{gaines2018,
  author  = {Gaines, B.~R. and Kim, J. and Zhou, H.},
  title   = {Algorithms for fitting the constrained lasso},
  journal = {Journal of Computational and Graphical Statistics},
  year    = {2018},
  volume  = {27},
  number  = {4},
  pages   = {861--871},
}

@article{gasssaaty1955,
  author  = {Gass, S.~I. and Saaty, T.~L.},
  title   = {The computational algorithm for the parametric objective function},
  journal = {Naval Research Logistics Quarterly},
  year    = {1955},
  volume  = {2},
  number  = {1--2},
  pages   = {39--45},
}

@article{hastie2004svm,
  author  = {Hastie, T. and Rosset, S. and Tibshirani, R. and Zhu, J.},
  title   = {The entire regularization path for the support vector machine},
  journal = {Journal of Machine Learning Research},
  year    = {2004},
  volume  = {5},
  pages   = {1391--1415},
}

@misc{hobson2026classo,
  author    = {Hobson, M. and Butcher, L. and Lentati, L.},
  title     = {{CLASSO}: a practical primal path-following algorithm for the constrained {LASSO}},
  year      = {2026},
  note      = {arXiv:2610.17024},
}

@misc{karush1939,
  author    = {Karush, W.},
  title     = {Minima of functions of several variables with inequalities as side conditions},
  note      = {M.Sc. thesis, Department of Mathematics, University of Chicago},
  year      = {1939},
}

@inproceedings{kuhntucker1951,
  author    = {Kuhn, H.~W. and Tucker, A.~W.},
  title     = {Nonlinear programming},
  editor    = {Neyman, J.},
  booktitle = {Proceedings of the Second Berkeley Symposium on Mathematical Statistics and Probability},
  pages     = {481--492},
  publisher = {University of California Press},
  address   = {Berkeley},
  year      = {1951},
}

@article{lee2016,
  author  = {Lee, J.~D. and Sun, D.~L. and Sun, Y. and Taylor, J.~E.},
  title   = {Exact post-selection inference, with application to the lasso},
  journal = {The Annals of Statistics},
  year    = {2016},
  volume  = {44},
  number  = {3},
  pages   = {907--927},
}

@inproceedings{mairal2012,
  author    = {Mairal, J. and Yu, B.},
  title     = {Complexity analysis of the lasso regularization path},
  booktitle = {Proceedings of the 29th International Conference on Machine Learning (ICML)},
  year      = {2012},
}

@article{markowitz1952,
  author  = {Markowitz, H.~M.},
  title   = {Portfolio selection},
  journal = {The Journal of Finance},
  year    = {1952},
  volume  = {7},
  number  = {1},
  pages   = {77--91},
}

@article{markowitz1956,
  author  = {Markowitz, H.~M.},
  title   = {The optimization of a quadratic function subject to linear constraints},
  journal = {Naval Research Logistics Quarterly},
  year    = {1956},
  volume  = {3},
  number  = {1--2},
  pages   = {111--133},
}

@article{houthakker1953,
  author  = {Houthakker, H.~S.},
  title   = {La forme des courbes d'{E}ngel},
  journal = {Cahiers du S\'eminaire d'\'Econom\'etrie},
  year    = {1953},
  volume  = {2},
  pages   = {59--66},
}

@book{dorfman1951,
  author    = {Dorfman, R.},
  title     = {Application of Linear Programming to the Theory of the Firm},
  publisher = {University of California Press},
  address   = {Berkeley},
  year      = {1951},
}

@techreport{markowitz1955,
  author      = {Markowitz, H.~M.},
  title       = {The optimization of quadratic functions subject to linear constraints},
  institution = {The RAND Corporation},
  number      = {RM-1438},
  year        = {1955},
  note        = {Research Memorandum, 21 February 1955; the working paper behind
                 \citet{markowitz1956}},
}

@book{markowitz1959,
  author    = {Markowitz, H.~M.},
  title     = {Portfolio Selection: Efficient Diversification of Investments},
  publisher = {Wiley},
  address   = {New York},
  year      = {1959},
}

@incollection{markowitz2021,
  author    = {Markowitz, H.~M. and Starer, D. and Fram, H. and Gerber, S.},
  title     = {Avoiding the downside: a practical review of the critical line algorithm for mean--semivariance portfolio optimization},
  booktitle = {Handbook of Applied Investment Research},
  chapter   = {17},
  publisher = {World Scientific},
  year      = {2021},
  doi       = {10.1142/9789811222634\_0017},
  url       = {https://doi.org/10.1142/9789811222634_0017},
}

@incollection{niedermayer2010,
  author    = {Niedermayer, A. and Niedermayer, D.},
  title     = {Applying {Markowitz's} critical line algorithm},
  booktitle = {Handbook of Portfolio Construction},
  pages     = {383--400},
  publisher = {Springer},
  year      = {2010},
}

@book{nocedal2006,
  author    = {Nocedal, J. and Wright, S.~J.},
  title     = {Numerical Optimization},
  edition   = {2nd},
  publisher = {Springer},
  address   = {New York},
  year      = {2006},
}

@misc{schmelzer2026constrained,
  author    = {Schmelzer, T.},
  title     = {The constrained {LASSO} is a portfolio: a numerical check of the {CLA}--{LASSO} identity under constraints},
  year      = {2026},
  note      = {Preprint. Source at \url{https://github.com/Jebel-Quant/lasso}},
}

@article{osborne2000,
  author  = {Osborne, M.~R. and Presnell, B. and Turlach, B.~A.},
  title   = {A new approach to variable selection in least squares problems},
  journal = {IMA Journal of Numerical Analysis},
  year    = {2000},
  volume  = {20},
  number  = {3},
  pages   = {389--403},
}

@article{perold1984,
  author  = {Perold, A.~F.},
  title   = {Large-scale portfolio optimization},
  journal = {Management Science},
  year    = {1984},
  volume  = {30},
  number  = {10},
  pages   = {1143--1160},
}

@article{roll2008,
  author  = {Roll, J.},
  title   = {Piecewise linear solution paths with application to direct weight optimization},
  journal = {Automatica},
  year    = {2008},
  volume  = {44},
  number  = {11},
  pages   = {2745--2753},
}

@article{rosset2007,
  author  = {Rosset, S. and Zhu, J.},
  title   = {Piecewise linear regularized solution paths},
  journal = {The Annals of Statistics},
  year    = {2007},
  volume  = {35},
  number  = {3},
  pages   = {1012--1030},
}

@article{ritter1962,
  author  = {Ritter, K.},
  title   = {Ein {V}erfahren zur {L}\"osung parameterabh\"angiger, nichtlinearer
             {M}aximum-{P}robleme},
  journal = {Unternehmensforschung},
  year    = {1962},
  volume  = {6},
  pages   = {149--166},
  doi     = {10.1007/BF01920852},
  note    = {Translated as Ritter (1967)},
}

@article{ritter1967,
  author  = {Ritter, K.},
  title   = {A method for solving nonlinear maximum-problems depending on parameters},
  journal = {Naval Research Logistics Quarterly},
  year    = {1967},
  volume  = {14},
  number  = {2},
  pages   = {147--162},
  note    = {English translation, by M. Meyer, of Ritter (1962)},
}

@misc{cvxcla,
  author       = {Schmelzer, T.},
  title        = {\texttt{cvxcla}: the Critical Line Algorithm with a covariance-operator interface},
  year         = {2026},
  publisher    = {Zenodo},
  doi          = {10.5281/zenodo.22209209},
  note         = {Version 2.0.0, open-source software (MIT licence).
                  Source at \url{https://github.com/cvxgrp/cvxcla}},
}

@article{sharpe1963,
  author  = {Sharpe, W.~F.},
  title   = {A simplified model for portfolio analysis},
  journal = {Management Science},
  year    = {1963},
  volume  = {9},
  number  = {2},
  pages   = {277--293},
}

@article{tibshirani1996,
  author  = {Tibshirani, R.},
  title   = {Regression shrinkage and selection via the lasso},
  journal = {Journal of the Royal Statistical Society, Series B},
  year    = {1996},
  volume  = {58},
  number  = {1},
  pages   = {267--288},
}

@article{tibshirani2013,
  author  = {Tibshirani, R.~J.},
  title   = {The lasso problem and uniqueness},
  journal = {Electronic Journal of Statistics},
  year    = {2013},
  volume  = {7},
  pages   = {1456--1490},
}

@article{tibshirani2011,
  author  = {Tibshirani, R.~J. and Taylor, J.},
  title   = {The solution path of the generalized lasso},
  journal = {The Annals of Statistics},
  year    = {2011},
  volume  = {39},
  number  = {3},
  pages   = {1335--1371},
}

@article{tibshirani2012df,
  author  = {Tibshirani, R.~J. and Taylor, J.},
  title   = {Degrees of freedom in lasso problems},
  journal = {The Annals of Statistics},
  year    = {2012},
  volume  = {40},
  number  = {2},
  pages   = {1198--1232},
}

@article{tibshirani2016psi,
  author  = {Tibshirani, R.~J. and Taylor, J. and Lockhart, R. and Tibshirani, R.},
  title   = {Exact post-selection inference for sequential regression procedures},
  journal = {Journal of the American Statistical Association},
  year    = {2016},
  volume  = {111},
  number  = {514},
  pages   = {600--620},
}

@article{tondel2003,
  author  = {T{\o}ndel, P. and Johansen, T.~A. and Bemporad, A.},
  title   = {An algorithm for multi-parametric quadratic programming and explicit {MPC} solutions},
  journal = {Automatica},
  year    = {2003},
  volume  = {39},
  number  = {3},
  pages   = {489--497},
}

@article{wolfe1959,
  author  = {Wolfe, P.},
  title   = {The simplex method for quadratic programming},
  journal = {Econometrica},
  year    = {1959},
  volume  = {27},
  number  = {3},
  pages   = {382--398},
}

@article{zou2007df,
  author  = {Zou, H. and Hastie, T. and Tibshirani, R.},
  title   = {On the ``degrees of freedom'' of the lasso},
  journal = {The Annals of Statistics},
  year    = {2007},
  volume  = {35},
  number  = {5},
  pages   = {2173--2192},
}

@article{fan2012,
  author  = {Fan, J. and Zhang, J. and Yu, K.},
  title   = {Vast portfolio selection with gross-exposure constraints},
  journal = {Journal of the American Statistical Association},
  year    = {2012},
  volume  = {107},
  number  = {498},
  pages   = {592--606},
}

@article{kato2009,
  author  = {Kato, K.},
  title   = {On the degrees of freedom in shrinkage estimation},
  journal = {Journal of Multivariate Analysis},
  year    = {2009},
  volume  = {100},
  number  = {7},
  pages   = {1338--1352},
}

@incollection{best1996,
  author    = {Best, M.~J.},
  title     = {An algorithm for the solution of the parametric quadratic programming problem},
  booktitle = {Applied Mathematics and Parallel Computing: Festschrift for Klaus Ritter},
  editor    = {Fischer, H. and Riedm\"uller, B. and Sch\"affler, S.},
  publisher = {Physica-Verlag},
  address   = {Heidelberg},
  year      = {1996},
  pages     = {57--76},
}

@article{gartner2012,
  author  = {G\"artner, B. and Jaggi, M. and Maria, C.},
  title   = {An exponential lower bound on the complexity of regularization paths},
  journal = {Journal of Computational Geometry},
  year    = {2012},
  volume  = {3},
  number  = {1},
  pages   = {168--195},
}

@inproceedings{gartner2009,
  author    = {G\"artner, B. and Giesen, J. and Jaggi, M. and Welsch, T.},
  title     = {A combinatorial algorithm to compute regularization paths},
  booktitle = {arXiv:0903.4856},
  year      = {2009},
  note      = {A generic parametric-quadratic-programming path method, applied to the
               LASSO, support vector machines, Markowitz portfolio selection and
               conjoint analysis},
}

@incollection{cottle2010,
  author    = {Cottle, R. W. and Infanger, G.},
  title     = {Harry {M}arkowitz and the Early History of Quadratic Programming},
  booktitle = {Handbook of Portfolio Construction: Contemporary Applications of
               {M}arkowitz Techniques},
  editor    = {Guerard, J. B.},
  publisher = {Springer},
  address   = {Boston, MA},
  year      = {2010},
  pages     = {179--211},
}

\end{document}